\documentclass[12pt]{article}
\usepackage[utf8]{inputenc}
\usepackage{natbib}
\usepackage{blindtext}
\usepackage[T1]{fontenc}
\usepackage{amsfonts}
\usepackage[english]{babel}
\usepackage{amsthm}
\usepackage{amssymb}
\usepackage{relsize}
\usepackage{amsmath}
\usepackage{mathrsfs} 
\usepackage[parfill]{parskip}
\usepackage{bbm}
\usepackage{graphicx}
\usepackage{enumerate}
\usepackage{enumitem}
\usepackage{mathtools}
\usepackage{xcolor}
\usepackage[pdftex , colorlinks = true , citecolor = red , linkcolor = blue ]{hyperref}	
\usepackage{array}
\usepackage{booktabs}
\usepackage{dsfont}
\usepackage{subcaption}
\usepackage{multirow}
\usepackage{longtable}

\newtheorem{lemma}{Lemma}
\newtheorem{theorem}{Theorem}
\newtheorem{condition}{Condition}

\newtheorem{definition}{Definition}
\newtheorem{assumption}{Assumption}

\newcommand{\e}{\mathbb{E}}
\newcommand{\cov}{\mathrm{Cov}}
\newcommand{\tr}{\mathrm{tr}}
\newcommand{\argmax}{\mathrm{argmax}}
\newcommand{\vect}{\mathrm{vec}}

\newcommand\indep{\protect\mathpalette{\protect\independenT}{\perp}}
\def\independenT#1#2{\mathrel{\rlap{$#1#2$}\mkern2mu{#1#2}}}
\DeclareMathOperator*{\argmin}{argmin}
\DeclareMathOperator*{\card}{card}

\numberwithin{equation}{section}
\allowdisplaybreaks

\begin{document}
\title{\Large{\textbf{Joint Functional Gaussian Graphical Models}}}
\author{Ilias Moysidis\\
    imoysidi@fredhutch.org \\
	Public Health Sciences Division\\
	Fred Hutchinson Cancer Research Center\\
	Seattle, WA 98109-1024, USA\\
	Bing Li\\
	bxl9@psu.edu \\
	Department of Statistics\\
	Pennsylvania State University\\
	Pennsylvania, PA 16803, USA}
\date{}

\maketitle

\begin{abstract}
	Functional graphical models explore dependence relationships of random processes. This is achieved through estimating the precision matrix of the coefficients from the Karhunen-Loeve expansion. This paper deals with the problem of estimating functional graphs that consist of the same random processes and share some of the dependence structure. By estimating a single graph we would be shrouding the uniqueness of different sub groups within the data. By estimating a different graph for each sub group we would be dividing our sample size. Instead, we propose a method that allows joint estimation of the graphs while taking into account the intrinsic differences of each sub group. This is achieved by a hierarchical penalty that first penalizes on a common level and then on an individual level. We develop a computation method for our estimator that deals with the non-convex nature of the objective function. We compare the performance of our method with existing ones on a number of different simulated scenarios. We apply our method to an EEG data set that consists of an alcoholic and a non-alcoholic group, to construct brain networks. 
\end{abstract}

\section{Introduction}

Functional graphical models provide insight on the conditional dependence structure among the components of a multivariate random function. Datasets such as those arising from functional magnetic resonance imaging (fMRI) and electroencephalography (EEG) motivate research in this area. In particular, for the EEG dataset, a curve is recorded at each location of the brain and a network is constructed based on a sample of subjects each having a vector of curves. However, these datasets often consist of samples originating from different subpopulations that share behavioral or genetic characteristics. Such subpopulations could be ADHD and non-ADHD patients or alcoholic and non-alcoholic subjects. We can assume that the similarities between the groups translate into a common graph structure. Merging the data together to estimate a single graph would ignore the differences among subpopulations. Dividing the data to estimate a graph individually for each subpopulation would waste potential information in the common structure that lies across the different subpopulations. The goal of this paper is to develop a model that uses all the data to estimate the common structure, but also leaves room for differences between the graphs. We call this model the joint functional graphical model.

Consider a vector of random functions $\mathbf{g}=(g_{1},\ldots,g_{p})$. The graphical model of $\mathbf{g}$ is represented by an undirected graph $\mathcal{G}=(\mathcal{V},\mathcal{E})$, where $\mathcal{V}=\{1,\ldots,p\}$ is the set of nodes and $\mathcal{E}\subseteq\{(i,j):i,j\in\mathcal{V},i\ne j\}$ is the set of edges. For convenience we assume $i<j$ for $(i,j)\in\mathcal{E}$, because for an undirected graph $(i,j)$ and $(j,i)$ represent the same edge. The set of edges $\mathcal{E}$ is defined by the relation
\begin{align}
	(i,j)\notin\mathcal{E}\quad\Leftrightarrow\quad g_{i}\indep g_{j}|\mathbf{g}_{\backslash\{i,j\}}\label{rel:1}
\end{align}
where $\mathbf{g}_{\backslash\{i,j\}}$ represents $\mathbf{g}$ with its $i$-th and $j$-th components removed.

Functional graphical models have undergone dynamic development in the recent years. \cite{zhu2016bayesian} extended the notions of Markov distribution and hyper Markov laws to the functional case and developed a Bayesian approach for functional graphical models by proposing a hyper-inverse Wishart process prior for the covariance operator, and assuming the random elements are multivariate Gaussian processes. \cite{qiao2019functional} introduced the functional Gaussian graphical model (FGGM) where the random functions are assumed to be Gaussian random elements in a Hilbert space, and the network is constructed using an association of the relation \eqref{rel:1} with the coefficients of their Karhunen-Loeve expansion. \cite{soleali} proposed a semiparametric functional copula Gaussian graphical model for random functions that are not originally Gaussian processes, but for which there exist one-to-one transformations of their Karhunen-Loeve expansion coefficients that makes them Gaussian. \cite{li2018nonparametric} bypasses the Gaussian assumption altogether by replacing the conditional independence relationship with additive conditional independence. This approach allows for nonlinear or heteroscedastic relations between the random processes.

In the multivariate case, the research on covariance selection dates back to \cite{dempster1972covariance} who proposed backward selection, in which we start with a fully connected graph and at each step remove edges that are not significant according to a partial correlation test until all remaining edges are significant. Another approach to constructing graphical models is via sparse inducing penalty which has become popular because it can be applied to the case where the dimension of the graph is much bigger than the sample size. \cite{meinshausen2006high} reformulated the covariance selection problem to an $\ell_{1}$ penalized variable selection problem where each variable is regressed against the others to estimate its set of effective predictors. \cite{yuan2007model} propose an $\ell_{1}$ penalized loglikelihood method which has the advantage of merging variable selection and precision matrix estimation into one problem. 

The first paper that deals with the joint estimation problem in the multivariate case is by \cite{guo2011joint}. They decompose each precision matrix into a shared component and a unique component, and proposed a nonconvex hierarchical penalty that utilizes the information from all the data by penalizing first on a common and then on an individual level. \cite{danaher2014joint} developed the computationally attractive methods of Fused Graphical Lasso and Group Graphical Lasso that allow joint estimation while also have the advantage of employing convex penalties. The first forces a common structure using a penalty that promotes identical edge values across subpopulations, while the latter only propels shrinkage on a common level. A general framework for joint precision matrix estimation is provided in \cite{saegusa2016joint}. They use a graph Laplacian penalty that allows for different levels of similarity between subpopulations. Motivated by the high dimensionality of the datasets, \cite{cai2016joint} propose a computationally fast method for jointly estimating precision matrices, following the steps of the Dantzig Selector \citep{candes2007dantzig}.

In this paper we focus on the joint functional graphical model under the Gaussian assumption, and we call our method the joint functional Gaussian graphical model (JFGGM). We propose a nonconvex objective function that achieves regularization which encourages both a common graph structure and individual sparsity. To estimate the global minimizer we are using the local linear approximation (LLA) method, which is an algorithm that at each step proposes a convex, non-smooth objective function. To calculate the global minimizer at each step of the LLA algorithm we employ the alternating direction method of multipliers (ADMM) algorithm. Furthermore, we establish the asymptotic consistency of our method with overwhelming probability. In addition, we do a simulation study, where we compare our method with the FGGM applied on each subpopulation individually. Finally, we estimate the graph of the EEG dataset with the use of our method.

The rest of the paper is organized as follows. In section 2 we give an overview of the functional graphical lasso (fglasso) method, and introduce the JFGGM. In section 3 we develop the algorithmic procedure for estimating the JFGGM. In section 4 we establish the asymptotic consistency of our estimator. In section 5 we compare by simulation our method with the FGGM. In section 6 we apply JFGGM to the EEG dataset.

\section{Methodology}\label{methodology}

In this section we provide the theoretical background for the functional graphical model presented in \cite{qiao2019functional}. We, then, propose our method, the JFGGM, for joint estimation of graphical models that come from different subpopulations.

\subsection{Functional Graphical Models}\label{functional:graphical:models}
Let $(\mathbf{\Omega},\mathcal{F},P)$ be a probability space, let $T$ be an interval in $\mathbb{R}$, let $\mathbb{H}_{1},\ldots,\mathbb{H}_{p}$ be Hilbert spaces of real-valued functions on $T$. Let $g_{i}:\Omega\rightarrow\mathbb{H}_{i}$ be a random element in $\mathbb{H}_{i}$, so that $(g_{1},\ldots,g_{p})$ is a random element in the direct sum $\oplus_{i=1}^{p}\mathbb{H}_{i}$, which is the Cartesian product $\mathbb{H}_{1}\times\cdots\times\mathbb{H}_{p}$ together with the inner product $\langle f_{1},g_{1}\rangle_{\mathbb{H}_{1}}+\ldots+\langle f_{p},g_{p}\rangle_{\mathbb{H}_{p}}$.

The functional Gaussian graphical model (FGGM) developed by \cite{qiao2019functional} is under the assumption that $\mathbf{g}=(g_{1},\ldots,g_{p})$ is a multivariate Gaussian random element in $\mathbb{H}_{1}\times\cdots\times\mathbb{H}_{p}$. Each $\mathbb{H}_{i}$ is an $M$-dimensional subspace of the space of square integrable functions $\mathbb{L}_{2}(T,\mathcal{B}(T),\mu)$, where $\mathcal{B}(T)$ is the Borel $\sigma$-field generated by the open sets in $T$ and $\mu$ is the Lebesgue measure. Under this assumption, relation \eqref{rel:1} reduces to
\begin{align}\label{rel:2}
	(i,j)\notin\mathcal{E}\quad\Leftrightarrow\quad \cov[g_{i}(s),g_{j}(t)\,|\,g_{\backslash\{i,j\}}]=0\enspace\forall\, s,t\in T. 
\end{align}

Associated with the stochastic process $g_{j}$ is the mean function $h_{j}(t)=\e (g_{j}(t))$ and the covariance function $K_{j}(s,t)=\cov(g_{j}(s),g_{j}(t))$. Processes with well defined mean and covariance functions are called second-order processes. If such a process is mean-square continuous, then the Karhunen-Loeve expansion theorem applies. To move forward, we need the following assumption. 

\medskip
\begin{assumption}\label{assumption:1}
	Each $g_{j}$ is jointly measurable with respect to the product $\sigma$-field $\mathcal{B}(T)\times\mathcal{F}$ and $g_{j}(\cdot,\omega)\in\mathbb{H}_{j}$ for each $\omega\in\mathbf{\Omega}$. Furthermore, we have $\e\|\mathbf{g}\|^{2}<\infty$.
\end{assumption}
This assumption implies the existence of the covariance operator of the $j$-th process
\begin{align*}
	\mathscr{K}_{j}:\mathbb{H}_{j}\rightarrow\mathbb{H}_{j},\quad f\mapsto\int_{T}K_{j}(\cdot,t)f(t)dt,
\end{align*}
and its spectral decomposition. That is, there exist pairs $\{(\lambda_{jm},\phi_{jm})\}_{m=1}^{M}$, with $\lambda_{j1}\geq\ldots\geq\lambda_{jM}$, such that
\begin{align*}
	\int_{T}K_{j}(s,t)\phi_{jm}(t)dt=\lambda_{jm}\phi_{jm}(s),\quad s\in T,\quad m=1,\ldots,M,
\end{align*}
and
\begin{align*}
	\int_{T}\phi_{jm}(t)\phi_{jm'}(t)dt=\delta_{mm'},\quad m,m'=1,\ldots,M,
\end{align*}
where $\delta_{mm'}$ is the Kronecker $\delta$-function. To link the spectral decomposition of $\mathscr{K}_{j}$ with the Karhunen-Loeve expansion of $g_{j}$ we need the following assumption.

\medskip
\begin{assumption}\label{assumption:2}
	For each $j$, the mean and covariance functions of $g_{j}$ are continuous.
\end{assumption}
Without loss of generality we assume that $h_{j}=\mathbf{0}$. Under Assumptions~\ref{assumption:1} and \ref{assumption:2}, we have the following Karhunen-Loeve expansion for $g_{j}$:
\begin{align*}
	g_{j}(t)=\sum_{m=1}^{M}a_{jm}\phi_{jm}(t),
\end{align*}
where $a_{jm}=\int_{T}g_{j}(t)\phi_{jm}(t)dt$ is distributed as $\mathcal{N}(0,\lambda_{jm})$, and $a_{jm}\indep a_{jm'}$ for $m\ne m'$. A more detailed analysis of these matters can be found in \cite{hsing2015theoretical}.

Let $\mathbf{a}_{j}=(a_{j1},\ldots,a_{jM})^{\intercal}$. Since $\mathbf{g}$ is a mean zero Gaussian random element, the vector $\mathbf{a}=(\mathbf{a}_{1}^{\intercal},\ldots,\mathbf{a}_{p}^{\intercal})^{\intercal}$ follows a $pM$-dimensional normal distribution with mean $\mathbf{0}$ and covariance matrix $\mathbf{\Sigma}=(\mathbf{\Sigma}_{jl})_{p\times p}\in\mathbb{R}^{pM\times pM}$ where $\mathbf{\Sigma}_{jl}=\cov(\mathbf{a}_{j},\mathbf{a}_{l})\in\mathbb{R}^{M\times M}$. Let $\mathbf{\Omega}=\mathbf{\Sigma}^{-1}$ be the precision matrix. The following theorem \citep{qiao2019functional} links the conditional independence of random elements with the zero entries of $\mathbf{\Omega}$. 

\medskip
\begin{theorem}\label{theorem:1}
	Let $i,j\in\mathcal{V},\, i\ne j$, and let $\mathbf{\Omega}_{ij}$ be the $M\times M$ matrix corresponding to the $(i,j)$-th block matrix of $\mathbf{\Omega}$. Then,
	\begin{align*}
		g_{i}\indep g_{j}|\mathbf{g}_{\backslash\{i,j\}}\quad\Leftrightarrow\quad\mathbf{\Omega}_{ij}=\mathbf{0}.
	\end{align*}
\end{theorem}

The defining relation \eqref{rel:2} of the functional graph is difficult to work with. Theorem \ref{theorem:1} provides an equivalent condition relating it to the coefficients of the Karhunen-Loeve expansion, thus, reducing the functional setting to the multivariate setting, which has been extensively studied. 

Denote the estimator of $\mathbf{\Sigma}$ by $\hat{\mathbf{\Sigma}}$, whose precise definition is given in Section \ref{computation}, where we give detailed analysis of $\hat{\mathbf{\Sigma}}$ at the sample level. Based on Theorem \ref{theorem:1}, \cite{qiao2019functional} introduce the fglasso to estimate $\mathbf{\Omega}$ by
\begin{align*}
	\hat{\mathbf{\Omega}}=\argmin_{\mathbf{\Omega}\succ\mathbf{0}}\left\{\tr(\hat{\mathbf{\Sigma}}\mathbf{\Omega})-\log\det(\mathbf{\Omega})+\lambda\sum_{i\ne j}\|\mathbf{\Omega}_{ij}\|_{F}\right\},
\end{align*}
where $\|\cdot\|_{F}$ is the Frobenius norm. The loss function $\tr(\hat{\mathbf{\Sigma}}\mathbf{\Omega})-\log\det(\mathbf{\Omega})$ is the negative loglikelihood and is responsible for producing a precision matrix that is going to make the data most likely to be observed given the assumptions of the model, and the penalty function is responsible for enforcing sparsity on the estimator. Notice that to encourage blockwise sparsity for $\mathbf{\Omega}$, they use a groupwise penalty induced by the Frobenius norm.

\subsection{Joint Functional Graphical Models}\label{subsection:jfgm}
We now relax the assumption that all of the observed functional data $\mathbf{g}_{i}=(g_{i1},\ldots,g_{ip})^{\intercal}$ are realizations of the same Gaussian process. Instead, we assume that there are $K$ different subpopulations whose graphs, though not identical, share a set of common edges. Suppose, for each $k=1,\ldots,K$, we observe an i.i.d. sample $\mathbf{g}_{1}^{(k)},\ldots,\mathbf{g}_{n}^{(k)}$ of random elements in $\oplus_{j=1}^{p}\mathbb{H}_{j}$, where $\mathbf{g}_{i}^{(k)}=(g_{i1}^{(k)},\ldots,g_{ip}^{(k)})^{\intercal}$. We assume that $\mathbf{g}_{i}^{(k)}$ is a zero mean Gaussian random element. By performing Karhunen-Loeve expansion on each $g_{ij}^{(k)}$ we obtain the vector of coefficients $\mathbf{a}_{ij}^{(k)}=(a_{ij1}^{(k)},\ldots,a_{ijM}^{(k)})^{\intercal}$. Since each $\mathbf{g}_{i}^{(k)}$ is a mean zero Gaussian random element, the vector $\mathbf{a}_{i}^{(k)}=(\mathbf{a}_{i1}^{(k)\,\intercal},\ldots,\mathbf{a}_{ip}^{(k)\,\intercal})^{\intercal}$ is distributed as $\mathcal{N}_{pM}(\mathbf{0},\mathbf{\Sigma}^{(k)})$, for $i=1,\ldots,n,k=1,\ldots,K$. Let $\mathbf{\Omega}^{(k)}=\mathbf{\Sigma}^{(k)-1}$ be the precision matrix corresponding to the $k$-th subpopulation.

To take advantage of the information across the subpopulations we reparameterize each off-diagonal block $\mathbf{\Omega}_{jl}^{(k)}$ of $\mathbf{\Omega}^{(k)}$ as $\theta_{jl}\mathbf{\Gamma}_{jl}^{(k)}$, where $\theta_{jl}\in\mathbb{R}$ is common to all subpopulations and $\mathbf{\Gamma}_{jl}^{(k)}\in\mathbb{R}^{M\times M}$ is specific to subpopulation $k$. For identifiability, we assume $\theta_{jl}\geq0$. To preserve symmetry we require $\theta_{lj}=\theta_{jl}$ and $\mathbf{\Gamma}_{lj}^{(k)\,\intercal}=\mathbf{\Gamma}_{jl}^{(k)}$. For the diagonal-block matrices we require $\theta_{jj}=1$ and $\mathbf{\Gamma}_{jj}^{(k)}=\mathbf{\Omega}_{jj}^{(k)}$. In this reparameterization the common factor $\theta_{jl}$ controls the presence of the common edge between the nodes $j$ and $l$ in all of the graphs, and $\mathbf{\Gamma}_{jl}^{(k)}$ accommodates the differences between individual graphs. This reparameterization is similar to that of \cite{guo2011joint} with the difference that in our case the individual element is a matrix rather than a number. To estimate this model we propose to minimize
\begin{align}\label{obj:1}
	\sum_{k=1}^{K}\left[\tr\left(\hat{\mathbf{\Sigma}}^{(k)}\mathbf{\Omega}^{(k)}\right)-\log\det\left(\mathbf{\Omega}^{(k)}\right)\right]+\lambda_{1}\sum_{j\ne l}\theta_{jl}+\lambda_{2}\sum_{j\ne l}\sum_{k=1}^{K}\|\mathbf{\Gamma}_{jl}^{(k)}\|_{F}  
\end{align}
over all $\theta_{jl}$ and $\mathbf{\Gamma}_{jl}^{(k)}$ specified above. The first penalty function penalizes the common factors $\theta_{jl}$ and is responsible for identifying the zeros across all precision matrices. That is, if $\theta_{jl}$ is zero then there is no edge between the nodes $j$ and $l$ in all $K$ graphs. The second penalty function penalizes the individual factors $\mathbf{\Gamma}_{jl}^{(k)}$ and is responsible for identifying the zeros that are specific to each graph. That is, for a nonzero $\theta_{jl}$ some of the matrices $\mathbf{\Gamma}_{jl}^{(1)},\ldots,\mathbf{\Gamma}_{jl}^{(K)}$ can be zero, which means that the edge connecting the nodes $j$ and $l$ may be absent in some of the $K$ graphs but present in others. 

However, the objective function \eqref{obj:1} is hard to minimize because of its complexity. It includes two groups of variables over which we have to optimize, and two parameters that we have to tune. Additionally, due to the restrictions on the variables that ensure identifiability and positive definiteness, the domain of the objective function is not directly intuitive. As will be shown in Theorem~\ref{theorem:2}, \eqref{obj:1} is equivalent to the much simpler form
\begin{align}\label{obj:3}
	\sum_{k=1}^{K}\left[\tr\left(\hat{\mathbf{\Sigma}}^{(k)}\mathbf{\Omega}^{(k)}\right)-\log\det\left(\mathbf{\Omega}^{(k)}\right)\right]+2(\lambda_{1}\lambda_{2})^{1/2}\sum_{j\ne l}\left(\sum_{k=1}^{K}\|\mathbf{\Omega}_{jl}^{(k)}\|_{F}\right)^{1/2},
\end{align}
where $\mathbf{\Omega}^{(1)},\ldots,\mathbf{\Omega}^{(K)}$ are positive definite matrices.

Let $\mathbf{\Theta}=(\theta_{jl})$, $\mathbf{\Gamma}^{(k)}=(\mathbf{\Gamma}_{jl}^{(k)})$, $\mathbf{\Omega}^{(k)}=(\mathbf{\Omega}_{jl}^{(k)})$ for $\theta_{jl}, \mathbf{\Gamma}_{jl}^{(k)}, \mathbf{\Omega}_{jl}^{(k)}$ defined above, and define the lists $\mathbf{\Gamma}=(\mathbf{\Gamma}^{(1)},\ldots,\mathbf{\Gamma}^{(K)})$, $\mathbf{\Omega}=(\mathbf{\Omega}^{(1)},\ldots,\mathbf{\Omega}^{(K)})$, $\hat{\mathbf{\Gamma}}=(\hat{\mathbf{\Gamma}}^{(1)},\ldots,\hat{\mathbf{\Gamma}}^{(K)})$, $\hat{\mathbf{\Omega}}=(\hat{\mathbf{\Omega}}^{(1)},\ldots,\hat{\mathbf{\Omega}}^{(K)})$. The proof of the following theorem can be found in the appendix.
\medskip
\begin{theorem}\label{theorem:2}
	Let $(\hat{\mathbf{\Theta}},\hat{\mathbf{\Gamma}})$ be a local minimizer of \eqref{obj:1}. Then, there exists a local minimizer $\hat{\mathbf{\Omega}}$ of \eqref{obj:3} such that $\hat{\mathbf{\Omega}}_{jl}^{(k)}=\hat{\theta}_{jl}\hat{\mathbf{\Gamma}}_{jl}^{(k)}$ for all $j,l,k$. Conversely, let $\hat{\mathbf{\Omega}}$ be a local minimizer of \eqref{obj:3}. Then, there exists a local minimizer $(\hat{\mathbf{\Theta}},\hat{\mathbf{\Gamma}})$ of \eqref{obj:1} such that $\hat{\theta}_{jl}\hat{\mathbf{\Gamma}}_{jl}^{(k)}=\hat{\mathbf{\Omega}}_{jl}^{(k)}$ for all $j,l,k$. 
\end{theorem}

\section{Sample-Level Implementation}\label{computation}

The goal of this section is to develop a procedure to calculate the minimizer of \eqref{obj:3}. We begin by providing a formula for the computation of the sample covariance matrix $\hat{\mathbf{\Sigma}}$ mentioned at the end of subsection \ref{functional:graphical:models}. Furthermore, because the objective function \eqref{obj:3} is nonconvex, we instead optimize an approximate version of it. Finally, to calculate the minimizer of the approximate objective function we employ the ADMM algorithm because it provides closed form solutions for the updates, makes good use of R's vectorized operations, and has shown to outperform other popular algorithms on similar problems \citep{scheinberg2010sparse}.

\subsection{Sample Covariance Matrix Estimation}
Since we will be concerned with a fixed $k$ in this subsection, we drop the superscript $(k)$ for simplicity of notation. First, we need to estimate the eigenpairs $(\lambda_{jm},\phi_{jm})$. To do that we follow the procedure described in \cite{ramsay2001functional}. Assume that we observe each curve $g_{ij}$ at equally-spaced time points $t_{1}<\ldots<t_{\nu}$, where $t_{1}$ and $t_{\nu}$ are the endpoints of the interval $T$. Without loss of generality, we assume that the data are centered. That is $\sum_{i=1}^{n}g_{ij}(t_{q})=0$, for all $j,q$. 

As we will show, the procedure for estimating the eigenpairs of the covariance operator is very similar to the multivariate case. We start by providing a matrix approximation of the covariance operator at the population level. Then, we show how to perform eigenvalue decomposition on the matrix approximation. Finally, by replacing the population-level quantities with sample-level quantities in the approximation of the covariance operator, we obtain the estimators of the eigenpairs.

Define
\begin{align*}
	\mathbf{K}_{j}=\left(\begin{array}{ccc}
		K_{j}(t_{1},t_{1}) &\ldots &K_{j}(t_{1},t_{\nu})\\
		\vdots&\ddots&\vdots\\
		K_{j}(t_{\nu},t_{1})&\ldots&K_{j}(t_{\nu},t_{\nu})
	\end{array}\right).
\end{align*}
Let also $\boldsymbol{\phi}_{jm}=(\phi_{jm}(t_{1}),\ldots,\phi_{jm}(t_{\nu}))^{\intercal}$, and $w=(t_{\nu}-t_{1})/(\nu-1)$ be the gap between two adjacent time points. Then, for large $\nu$,
\begin{align*}
	\int_{T}K_{j}(t_{r},t)\phi_{jm}(t)dt\approx w\sum_{q=1}^{\nu}K_{j}(t_{r},t_{q})\phi_{jm}(t_{q}).
\end{align*}
Therefore, the integral equation
\begin{align*}
	\int_{T}K_{j}(s,t)\phi_{jm}(t)dt=\lambda_{jm}\phi_{jm}(s)
\end{align*}
can be approximated by
\begin{align*}
	w\mathbf{K}_{j}\boldsymbol{\phi}_{jm}=\lambda_{jm}\boldsymbol{\phi}_{jm}.
\end{align*}
Then $(\lambda_{jm},\boldsymbol{\phi}_{jm})=(w^{-1}\mathbf{u}_{m}^{\intercal}\mathbf{K}_{j}\mathbf{u}_{m},\mathbf{u}_{m})$, where $\mathbf{u}_{m}$ is the solution to the eigenvalue problem:
\begin{equation}\label{eigen-prob}
	\begin{split}
		&\text{maximize}\quad \mathbf{u}_{m}^{\intercal}\mathbf{K}_{j}\mathbf{u}_{m}\\
		&\text{subject to}\quad \|\mathbf{u}_{m}\|_{2}=1\quad\text{and}\quad \mathbf{u}_{l}^{\intercal}\mathbf{u}_{m}=0\text{ for }l<m.
	\end{split}
\end{equation}
Let $\mathbf{a}_{i}=(\mathbf{a}_{i1}^{\intercal},\ldots,\mathbf{a}_{ip}^{\intercal})^{\intercal}$ denote the sample principal component score vector, where $\mathbf{a}_{ij}=(a_{ij1},\ldots,a_{ijM})^{\intercal}$. Similarly, for large $\nu$,
\begin{align*}
	a_{ijm}=\int_{T}g_{ij}(t)\phi_{jm}(t)dt\approx w\sum_{q=1}^{\nu}g_{ij}(t_{q})\phi_{jm}(t_{q}).
\end{align*}

Let $\hat{K}_{j}(s,t)=n^{-1}\sum_{i=1}^{n}g_{ij}(s)g_{ij}(t)$ be the estimator of the covariance function $K_{j}(s,t)$, and let $\hat{\mathbf{K}}_{j}=(\hat{K}_{j}(t_{r},t_{q}))$. Define
\begin{align*}
	\mathbf{G}_{j}=\left(\begin{array}{ccc}
		g_{1j}(t_{1})&\ldots&g_{1j}(t_{\nu})\\
		\vdots&\ddots&\vdots\\
		g_{nj}(t_{1})&\ldots&g_{nj}(t_{\nu})
	\end{array}\right).
\end{align*}
Then $\hat{\mathbf{K}}_{j}=n^{-1}\mathbf{G}_{j}^{\intercal}\mathbf{G}_{j}$. The estimator $(\hat{\lambda}_{jm},\hat{\boldsymbol{\phi}}_{jm})$ of the eigenpair $(\lambda_{jm},\boldsymbol{\phi}_{jm})$ is the solution to \eqref{eigen-prob} with $\mathbf{K}_{j}=\hat{\mathbf{K}}_{j}$. Finally, the estimated principal component scores are given by
\begin{align*}
	\hat{a}_{ijm}=\int_{T}g_{ij}(t)\hat{\phi}_{jm}(t)dt\approx w\sum_{q=1}^{\nu}g_{ij}(t_{q})\hat{\phi}_{jm}(t_{q})
\end{align*}
with which we calculate the sample covariance matrix $\hat{\mathbf{\Sigma}}=n^{-1}\sum_{i=1}^{n}\hat{\mathbf{a}}_{i}\hat{\mathbf{a}}_{i}^{\intercal}$.

\subsection{Penalty linearization}
We now resume the use of the superscript $(k)$ as we are concerned with calculations across the $K$ subpopulations in this subsection. Writing $2(\lambda_{1}\lambda_{2})^{1/2}$ as $\lambda$, the objective function \eqref{obj:3} becomes
\begin{align}
	\sum_{k=1}^{K}\left[\tr\left(\hat{\mathbf{\Sigma}}^{(k)}\mathbf{\Omega}^{(k)}\right)-\log\det\left(\mathbf{\Omega}^{(k)}\right)\right]+\lambda\sum_{j\ne l}\left(\sum_{k=1}^{K}\|\mathbf{\Omega}_{jl}^{(k)}\|_{F}\right)^{1/2}. \label{obj:4}
\end{align}

Notice that because of the square root in the penalty function, \eqref{obj:4} is not convex. To tackle this issue we use the Local Linear Approximation (LLA) method developed in \cite{zou2008one}. Suppose that we are given an initial value $\hat{\mathbf{\Omega}}_{(0)}=(\hat{\mathbf{\Omega}}_{(0)}^{(1)},\ldots,\hat{\mathbf{\Omega}}_{(0)}^{(K)})$ that is close to the true value. They propose locally approximating the penalty function by a linear function
\begin{footnotesize}
	\begin{align*}
		\sum_{j\ne l}\left(\sum_{k=1}^{K}\|\mathbf{\Omega}_{jl}^{(k)}\|_{F}\right)^{1/2}&\approx \sum_{j\ne l}\left(\sum_{k=1}^{K}\|\hat{\mathbf{\Omega}}_{(0),jl}^{(k)}\|_{F}\right)^{1/2}+\sum_{j\ne l}\tau_{(0),jl}\left(\sum_{k=1}^{K}\|\mathbf{\Omega}_{jl}^{(k)}\|_{F}-\sum_{k=1}^{K}\|\hat{\mathbf{\Omega}}_{(0),jl}^{(k)}\|_{F}\right)
	\end{align*}
\end{footnotesize}
where $\tau_{(0),jl}=2^{-1}\left(\sum_{k=1}^{K}\|\hat{\mathbf{\Omega}}_{(0),jl}^{(k)}\|_{F}\right)^{-1/2}$. Thus, at the $t$-th iteration, problem \eqref{obj:4} is decomposed into $K$ individual optimization problems
\begin{align}
	\hat{\mathbf{\Omega}}_{(t)}^{(k)}=\argmin_{\mathbf{\Omega}\succ\mathbf{0}}\left\{\tr\left(\hat{\mathbf{\Sigma}}^{(k)}\mathbf{\Omega}\right)-\log\det(\mathbf{\Omega})+\lambda\sum_{j\ne l}\tau_{(t-1),jl}\|\mathbf{\Omega}_{jl}\|_{F}\right\},  \label{obj:5}
\end{align}
where $k=1,\ldots,K$. It is shown in the same paper, that if the initial value at iteration $t=0$ is reasonably good, then with only one iteration $(t=1)$ we can get a good sparse estimate. In our paper, for the initial value we use the precision matrices estimated separately using the functional graphical lasso of \cite{qiao2019functional} for each subpopulation.

\subsection{ADMM algorithm for optimization}\label{ADMM}
To solve \eqref{obj:5} we employ the ADMM algorithm as described in \cite{boyd2011distributed}. The primal problem is given by
\begin{align*}
	&\text{minimize}\quad\tr(\mathbf{S\Omega})-\log\det(\mathbf{\Omega})+\lambda\sum_{j\ne l}\tau_{jl}\|\mathbf{\Omega}_{jl}\|_{F}\\
	&\text{subject to}\quad\mathbf{\Omega}\succ\mathbf{0},
\end{align*}
while the dual problem is
\begin{equation}
	\begin{split}\label{dual}
		&\text{minimize}\quad\tr(\mathbf{S\Omega})-\log\det(\mathbf{\Omega})+\lambda\sum_{j\ne l}\tau_{jl}\|\mathbf{Z}_{jl}\|_{F}\\
		&\text{subject to }
		\mathbf{\Omega}\succ\mathbf{0},\, \mathbf{Z}\text{ is symmetric},\,\mathbf{\Omega}-\mathbf{Z}=\mathbf{0}.
	\end{split}
\end{equation}
Under certain conditions, the solutions of the primal and dual problems coincide. 

Define the augmented Lagrangian function
\begin{align*}
	L_{b}(\mathbf{\Omega},\mathbf{Z},V)=\tr(\mathbf{S\Omega})-\log\det(\mathbf{\Omega})+\lambda\sum_{j\ne l}\tau_{jl}\mathbf{Z}_{jl}+\tr[\mathbf{V}^{\intercal}(\mathbf{\Omega}-\mathbf{Z})]+\frac{b}{2}\|\mathbf{\Omega}-\mathbf{Z}\|_{F}^{2}.
\end{align*}
The ADMM algorithm provides iterative formulas that approach the solution of the dual problem. The iteration formulas for \eqref{dual} are given by
\begin{align*}
	\mathbf{\Omega}^{(t+1)}&=\argmin_{\mathbf{\Omega}}\frac{\partial L_{b}(\mathbf{\Omega},\mathbf{Z}^{(t)},\mathbf{V}^{(t)})}{\partial\mathbf{\Omega}}=\mathbf{Y}\left\{\frac{1}{2}\left[\mathbf{\Lambda}+\left(\mathbf{\Lambda}^{2}+\frac{4}{b}\mathbf{I}_{pM}\right)^{1/2}\right]\right\}\mathbf{Y}^{\intercal}\\
	\mathbf{Z}_{jl}^{(t+1)}&=\argmin_{\mathbf{Z}}\frac{\partial L_{b}(\mathbf{\Omega}^{(t+1)},\mathbf{Z},\mathbf{V}^{(t)})}{\partial\mathbf{Z}_{jl}}\\
	&=
	\begin{cases}
		\mathbf{\Omega}_{jj}^{(t+1)}+\frac{1}{b}\mathbf{V}_{jj}^{(t)},&\quad\text{if }j=l\\
		\max\left(\mathbf{0},1-\frac{\lambda\tau_{jl}}{b\|\mathbf{\Omega}_{jl}^{(t+1)}+\frac{1}{b}\mathbf{V}_{jl}^{(t)}\|_{F}}\right)(\mathbf{\Omega}_{jl}^{(t+1)}+\frac{1}{b}\mathbf{V}_{jl}^{(t)}),&\quad\text{if }j\ne l
	\end{cases}\\
	\mathbf{V}^{(t+1)}&=\mathbf{V}^{(t)}+b(\mathbf{\Omega}^{(t+1)}-\mathbf{Z}^{(t+1)}),
\end{align*}
where $\mathbf{Y}$ is a matrix whose columns are the eigenvectors, $\mathbf{\Lambda}$ is a diagonal matrix of the eigenvalues, obtained by performing eigenvalue decomposition on 
\begin{align*}
	\mathbf{Z}^{(t)}-\frac{1}{b}(\mathbf{S}+\mathbf{V}^{(t)}),
\end{align*}
and $b$ is a positive constant that affects the convergence speed and accuracy of the algorithm. The initial values for $\mathbf{Z}$ and $\mathbf{V}$ are $\mathbf{1}_{pM}\mathbf{1}_{pM}^{\intercal}$ and $\mathbf
{0}_{pM}$, respectively. From the formulas of the updates, we can see that the computational complexity for the ADMM algorithm at each iteration is that of an eigenvalue decomposition of a matrix in $\mathbb{R}^{pM\times pM}$, which is $\mathcal{O}(p^{3}M^{3})$.

When it comes to estimating the graph, we do not use a threshold to determine the zero components of $\mathbf{\Omega}$ produced by the ADMM. Instead, we use the zero entries of the dual variable $\mathbf{Z}$ produced by the ADMM, for two reasons. First, as it can be seen by the iteration formula, $\mathbf{Z}$ is a thresholding rule. Second, the ADMM theory states that as the number of iterations increase, the updates of $\mathbf{\Omega}$ and $\mathbf{Z}$ converge to the same point.

\section{Asymptotics}\label{section:asymptotics}
In this section we prove the asymptotic consistency of the one-step version of \eqref{obj:5}, that is, the estimator produced by the first step of the LLA algorithm uncovers the true graph with probability tending to 1. 

The asymptotic consistency of the fglasso estimator was established in \cite{qiao2019functional}. The main difference of our setting is twofold. First, we do not make a distinction between the truncated and true process, as we assume that the dimension of the Hilbert space is known for a given sample size $n$. Second, we have to take into account the weights $\tau_{jl}$ that accommodate the common structure, which adds an extra layer of complexity.

We denote the true precision matrix by $\mathbf{\Omega}_{0}^{(k)}=(\mathbf{\Sigma}_{0}^{(k)})^{-1}$, the number of principal components by $M_{n}$, the number of stochastic processes by $p_{n}$, and define the degree of the graph
\begin{align*}
	d_{n}^{(k)}=\max_{j=1,\ldots,p_{n}}\card\left(\left\{l:l\ne j,\mathbf{\Omega}_{0,jl}^{(k)}\ne \mathbf{0}\right\}\right),
\end{align*}
where $\card$ denotes the cardinality of the set. In our framework we assume that all three quantities $M_{n},p_{n},d_{n}^{(k)}$ diverge to infinity.

Let $\mathbf{A}$ be the block-matrix $(\mathbf{A}_{jl})$, with $\mathbf{A}_{jl}\in\mathbb{R}^{M\times M}$. Define the $M$-block versions of the $\ell_{\infty}$-matrix norm, the $\ell_{\infty}$-vector norm, and the $\ell_{1}$-matrix norm to be
\begin{align*}
	\|\mathbf{A}\|_{\infty}^{(M)}&=\max_{j=1,\ldots,p}\sum_{l=1}^{p}\|\mathbf{A}_{jl}\|_{F},\\
	\|\mathbf{A}\|_{\max}^{(M)}&=\max_{1\leq j,l\leq p}\|\mathbf{A}_{jl}\|_{F},\\
	\|\mathbf{A}\|_{1}^{(M)}&=\max_{l=1,\ldots,p}\sum_{j=1}^{p}\|\mathbf{A}_{jl}\|_{F},
\end{align*}
respectively. Similar block versions of these norms are also going to be used for various block matrices and vectors, and their definition will be implied. For two sequence of real numbers $a_{n},b_{n}$ we denote $a_{n}\asymp b_{n}$ if there exist positive constants $c_{1},c_{2}$ such that $c_{1}\leq|a_{n}|/|b_{n}|\leq c_{2}$ for all $n$.

Let $\mathcal{B}=\{(i_{1},j_{11}),\ldots,(i_{1},j_{1r_{1}}),\ldots,(i_{q},j_{q1}),\ldots,(i_{1},j_{qr_{q}})\}$, such that $i_{1}<\ldots<i_{q}$, and $j_{sm}<j_{sl}$ for all $s$ and $m<l$. We define
\begin{align*}
	\mathbf{A}_{\mathcal{B}}=\left(\vect(\mathbf{A}_{i_{1}j_{11}})^{\intercal},\ldots,\vect(\mathbf{A}_{i_{1}j_{1r_{1}}})^{\intercal},\ldots,\vect(\mathbf{A}_{i_{q}j_{q1}})^{\intercal},\ldots,\vect(\mathbf{A}_{i_{q}j_{qr_{q}}})^{\intercal}\right)^{\intercal}.
\end{align*}
Let $\mathbf{\Gamma}^{(k)}=(\mathbf{\Omega}_{0}^{(k)})^{-1}\otimes(\mathbf{\Omega}_{0}^{(k)})^{-1}$. We use $\mathbf{\Gamma}_{\mathcal{B}\mathcal{C}}^{(k)}\in\mathbb{R}^{M^{2}|\mathcal{B}|\times M^{2}|\mathcal{C}|}$ to denote the submatrix of $\mathbf{\Gamma}^{(k)}$ with blocks $\mathbf{\Gamma}_{(i,j),(m,l)}^{(k)}\in\mathbb{R}^{M^{2}\times M^{2}}$, where $(i,j)\in\mathcal{B}$ and $(m,l)\in\mathcal{C}$. To construct the matrix $\mathbf{\Gamma}_{\mathcal{B}\mathcal{C}}^{(k)}$ we first fix the coordinates $(i,j)$ to locate the block $\cov(\mathbf{a}_{i}^{(k)},\mathbf{a}_{j}^{(k)})\otimes\mathbf{\Sigma}_{0}^{(k)}$ and then we fix the coordinates $(m,l)$ to locate the block $\cov(\mathbf{a}_{i}^{(k)},\mathbf{a}_{j}^{(k)})\otimes\cov(\mathbf{a}_{m}^{(k)},\mathbf{a}_{l}^{(k)})$. For a set $D$, we denote by $D^{c}$ its complement. Let $\mathcal{S}^{(k)}=\mathcal{E}^{(k)}\cup\{(1,1),\ldots,(p,p)\}$, where $\mathcal{E}^{(k)}=\{(j,l)\,:\,\mathbf{\Omega}_{0,jl}^{(k)}\ne \mathbf{0}\}$. Define the quantities
\begin{align}
	C_{\mathbf{\Sigma}}^{(k)}=\|(\mathbf{\Omega}_{0}^{(k)})^{-1}\|_{\infty}^{(M)},\quad  C_{\mathbf{\Gamma}}^{(k)}=\|(\mathbf{\Gamma}_{\mathcal{S}\mathcal{S}}^{(k)})^{-1}\|_{\infty}^{(M)},\quad C_{\mathbf{\Gamma}^{2}}^{(k)}=\|(\mathbf{\Gamma}_{\mathcal{S}^{c}\mathcal{S}}^{(k)}\mathbf{\Gamma}_{\mathcal{S}\mathcal{S}}^{(k)})^{-1}\|_{\infty}^{(M^{2})}.
\end{align}

We first need to find conditions to establish concentration bounds for all entries of $\hat{\mathbf{\Sigma}}^{(k)}-\mathbf{\Sigma}_{0}^{(k)}$. To do so we adopt the same conditions as in \cite{qiao2019functional}.

\medskip
\begin{condition}\label{condition:concetration}
	$(i)$ The number of principal components, $M_{n}$, satisfies $M_{n}\asymp n^{\alpha}$ with some constant $\alpha\geq 0$; $(ii)$ For each $j\in\mathcal{V}$, the eigenvalue sequence $\{\lambda_{jm}^{(k)}\}_{m=1}^{M_{n}}$ is decreasing; $(iii)$ There exists some constant $\beta >1$ with $\alpha\beta<1/4$ such that $\lambda_{jm}^{(k)}\asymp m^{-\beta}$ and $d_{jm}^{(k)}\lambda_{jm}^{(k)}=\mathcal{O}(m)$ for each $m=1,\ldots,M_{n}$ and $j\in \mathcal{V}$, where $d_{jm}^{(k)}=2\sqrt{2}\{(\lambda_{j(m-1)}^{(k)}-\lambda_{jm}^{(k)})^{-1},(\lambda_{jm}^{(k)}-\lambda_{j(m+1)}^{(k)})^{-1}\}$.
\end{condition}

Parameter $\alpha$ controls the dimension of the Hilbert spaces, while parameter $\beta$ determines how fast the eigenvalues $\lambda_{jm}^{(k)}$ and eigengaps $\lambda_{jm}^{(k)}-\lambda_{j(m-1)}^{(k)}$ converge to zero. We also need a condition for the weights $\tau_{jl}$ in order to establish the optimality of our estimator.
\medskip
\begin{condition}\label{condition:weights}
	For any $\gamma>2$, there exist positive numbers $a_{1}^{(k)},a_{2}^{(k)}$ such that
	$a_{1}^{(k)}>a_{2}^{(k)} C_{\mathbf{\Gamma}^{2}}^{(k)}$ and
	\begin{align*}
		\min_{(j,l)\in
			\mathcal{S}^{(k)c}}\tau_{jl}> a_{1}^{(k)},\quad\max_{(j,l)\in\mathcal{S}^{(k)}}\tau_{jl}< a_{2}^{(k)},
	\end{align*}
	with probability greater than $1-(M_{n}p_{n})^{2-\gamma}$ each.
\end{condition}
To gain intuition about this condition, let us assume that the initial estimator $\hat{\mathbf{\Omega}}_{(0)}$ is not far from the truth. In the best case scenario, where the edge sets $\mathcal{S}^{(k)}$ of all the subpopulations are identical and equal to $\mathcal{A}$, Condition \ref{condition:weights} can be seen as a minimum and maximum signal strength for the existent and absent edges respectively.
\begin{align*}
	&(j,l)\notin\mathcal{A}\quad\Rightarrow\quad\sum_{k=1}^{K}\|\hat{\mathbf{\Omega}}_{(0),jl}^{(k)}\|_{F}\rightarrow 0\quad\Rightarrow\quad\tau_{jl}\rightarrow\infty,\\
	&(j,l)\in\mathcal{A}\quad\Rightarrow\quad\sum_{k=1}^{K}\|\hat{\mathbf{\Omega}}_{(0),jl}^{(k)}\|_{F}\nrightarrow 0\quad\Rightarrow\quad\tau_{jl}\nrightarrow \infty.
\end{align*}		
The richer the diversity of the graphs, the harder it is for this condition to hold.

We are now ready to prove graph selection consistency for each subpopulation $k=1,\ldots,K$.

\medskip
\begin{theorem}\label{theorem:3}
	Suppose Conditions \ref{condition:concetration} and \ref{condition:weights} hold, $\gamma>2$,
	\begin{align*}
		\lambda_{n}=\frac{2(1+ C_{\mathbf{\Gamma}^{2}}^{(k)})M_{n}}{a_{1}-a_{2} C_{\mathbf{\Gamma}^{2}}^{(k)}}\sqrt{\frac{\log C_{2}+\gamma\log(M_{n}p_{n})}{C_{1}n^{1-2\alpha(1+\beta)}}},
	\end{align*}   
	\begin{align*}
		\min_{(j,l)\in\mathcal{S}^{(k)}}\|\mathbf{\Omega}_{0,jl}^{(k)}\|_{F}>\min\left\{\frac{1}{3C_{\mathbf{\Sigma}}^{(k)}d_{n}^{(k)}},\frac{1}{3(C_{\mathbf{\Sigma}}^{(k)})^{3} C_{\mathbf{\Gamma}}^{(k)}d_{n}^{(k)}}\right\}.
	\end{align*}
	Then, for all $n$ satisfying the lower bound
	\begin{scriptsize}
		\begin{align*}
			n^{1-2\alpha(1+\beta)}>\frac{\log[C_{2}(M_{n}p_{n})^{\gamma}]}{C_{1}}\max\left\{\frac{1}{C_{1}},\frac{2M_{n} C_{\mathbf{\Gamma}}^{(k)}\left[1+\frac{2a_{2}(1+ C_{\mathbf{\Gamma}^{2}}^{(k)})}{a_{1}-a_{2} C_{\mathbf{\Gamma}^{2}}^{(k)}}\right]}{\min\left\{\frac{1}{3C_{\mathbf{\Sigma}}^{(k)}d_{n}^{(k)}},\frac{1}{3(C_{\mathbf{\Sigma}}^{(k)})^{3} C_{\mathbf{\Gamma}}^{(k)}d_{n}^{(k)}}\right\}},6M_{n}d_{n}^{(k)} (C_{\mathbf{\Gamma}}^{(k)})^{2} C_{\mathbf{\Gamma}^{2}}^{(k)}\left[1+\frac{2a_{2}(1+ C_{\mathbf{\Gamma}^{2}}^{(k)})}{a_{1}-a_{2} C_{\mathbf{\Gamma}^{2}}^{(k)}}\right]^{2}\right\}^{2},
		\end{align*}
	\end{scriptsize}
	with probability greater than $1-3(M_{n}p_{n})^{2-\gamma}$ we have:
	\begin{enumerate}		
		\item $\|\hat{\mathbf{\Omega}}^{(k)}-\mathbf{\Omega}_{0}^{(k)}\|_{\text{max}}^{(M)}\leq \min\left\{\frac{1}{3C_{\mathbf{\Sigma}}^{(k)}d_{n}},\frac{1}{3(C_{\mathbf{\Sigma}}^{(k)})^{3} C_{\mathbf{\Gamma}}^{(k)}d_{n}}\right\}$,
		
		\item $\hat{\mathcal{S}}^{(k)}=\mathcal{S}^{(k)}$.
	\end{enumerate}
\end{theorem}

\section{Simulations}\label{simulations}
In this section we use simulation to compare the performance of the JFGGM with the FGGM applied separately to each subpopulation. To generate the data, we first construct the edge sets for all subpopulations, and then form the precision matrices described in section 2. 

To form the edge sets $\mathcal{E}^{(1)},\ldots,\mathcal{E}^{(K)}$, we follow two steps:
\begin{enumerate}
	\item Randomly choose a set of pairs $(j,l)\in\mathcal{V}\times\mathcal{V}\,,\,j<l$,
	as a percentage $s$ of the total number of edges $\binom{p}{2}$. This set constitutes the common graphical structure of the $K$ subpopulations and is denoted by $\mathcal{A}$.
	
	\item For each subpopulation $k$, randomly choose a set of pairs as a percentage $\rho$ of the number of common edges and denote this set by $\mathcal{B}^{(k)}$. The sets $\mathcal{B}^{(1)},\ldots,\mathcal{B}^{(K)}$ must satisfy
	\begin{align*}
		&\bigcup_{k=1}^{K}\mathcal{B}^{(k)}\cap\mathcal{A}=\varnothing,\\
		\text{and}\quad&\bigcap_{k=1}^{K}\mathcal{B}^{(k)}=\varnothing.
	\end{align*}
	These sets are the individual edge structure of each subpopulation. Combining the above, we define $\mathcal{E}^{(k)}=\mathcal{A}\cup \mathcal{B}^{(k)},\, k=1,\ldots,K$.
\end{enumerate} 

To form the precision matrices $\mathbf{\Omega}^{(1)},\ldots,\mathbf{\Omega}^{(K)}$, we follow three steps:
\begin{enumerate}
	\item Generate $a_{jl}^{(k)}$ independently from $\mathcal{U}(0,1)$ and form $\mathbf{A}^{(k)}=(\mathbf{A}_{jl}^{(k)}),\,k=1,\ldots,K$, by
	\begin{align*}
		\mathbf{A}_{jl}^{(k)}=\begin{cases}
			a_{jl}^{(k)}\mathbf{I}_{M},&\quad(j,l)\in\mathcal{E}^{(k)}\\
			\mathbf{I}_{M},&\quad j=l\\
			\mathbf{0},&\quad\text{otherwise}
		\end{cases},
	\end{align*}
	
	\item To ensure symmetry, let
	\begin{align*}
		\mathbf{B}^{(k)}=\frac{\mathbf{A}^{(k)}+\mathbf{A}^{(k)\,\intercal}}{2}.
	\end{align*}
	Note that by construction, the diagonal elements of $\mathbf{B}^{(k)}$ are 1.
	
	\item Let $b_{rs}^{(k)}$, $r=1,\ldots,pM,\,s=1,\ldots,pM$, be the $(r,s)$-th element of $\mathbf{B}^{(k)}$. To ensure positive definiteness, we use Gershgorin's Circle Theorem \citep{bell1965gershgorin} to define the precision matrices $\mathbf{\Omega}^{(k)}=(\omega_{rs}^{(k)}),\,k=1,\ldots,K$, such that
	\begin{align*}
		\omega_{rs}^{(k)}=
		\begin{cases}
			\frac{b_{rs}^{(k)}}{\sum_{q\ne r}|b_{rq}^{(k)}|},&\quad r\ne s\quad\text{and}\quad\sum_{q\ne r}|b_{rq}^{(k)}|> 0\\
			0,&\quad r\ne s\quad\text{and}\quad\sum_{q\ne r}|b_{rq}^{(k)}|=0\\
			1,&\quad r=s
		\end{cases}.
	\end{align*}
\end{enumerate}

With $\mathbf{\Omega}^{(1)},\ldots,\mathbf{\Omega}^{(K)}$ thus constructed, we are now ready to generate the observed data for each subpopulation. To do so, we follow three steps:

\begin{enumerate}
	\item Choose a basis $\boldsymbol{\phi}_{j}^{(k)}=(\phi_{j1}^{(k)},\ldots,\phi_{jM}^{(k)})^{\intercal}$, for all $j,k$.
	
	\item Generate $\mathbf{a}_{i}^{(k)}=(\mathbf{a}_{i1}^{(k)\,\intercal},\ldots,\mathbf{a}_{ip}^{(k)\,\intercal})^{\intercal}$ from a $\mathcal{N}_{pM}(\mathbf{0},(\mathbf{\Omega}^{(k)})^{-1})$ distribution, for $i=1,\ldots,n$ and all $k=1,\ldots,K$.
	
	\item Create the observed data $\mathbf{h}_{i}^{(k)}=(\mathbf{h}_{i1}^{(k)\,\intercal},\ldots,\mathbf{h}_{ip}^{(k)\,\intercal})^{\intercal}$, where
	\begin{align*}
		\mathbf{h}_{ij}^{(k)}(t)=\mathbf{a}_{ij}^{(k)\,\intercal}\boldsymbol{\phi}_{j}^{(k)}(t)+\epsilon_{ijt}^{(k)},
	\end{align*}
	and $\epsilon_{ijt}^{(k)}$ are i.i.d. $\mathcal{N}(0,\sigma^{2})$, for all $i,j,k,t$.
\end{enumerate}

We compare the JFGGM with the separate estimation with the FGGM on 12 different scenarios, which consist of all combinations of the variables $n=100$ and $200$, $p=80$ and $100$, $\rho=0,0.5$ and $1$. In all settings, the common structure of the $K=3$ subpopulations consists of $s=5\%$ of all possible edges $\binom{p}{2}$. The basis for the functional data, for each population, is $\{1,\sin t , \cos t\}$. Thus $M=3$. The variance of the error is $\sigma^{2}=0.05$, the number of time points is $\nu=100$, starting from 0 and ending at 1.

Receiver operating characteristic (ROC) curves are used to evaluate the performance of the two competing methods. For these curves we plot the average proportion of correctly detected links (ATPR) against the average proportion of falsely detected links (AFPR), over a range of values of $\lambda$. In particular,
\begin{align*}
	\text{AFPR}(\lambda)&=\frac{1}{K}\sum_{k=1}^{K}\frac{\sum_{1 \leq j < l \leq p}\mathds{1} \bigg (\mathbf{\Omega}_{0,jl}^{(k)}=\mathbf{0}\,,\, \hat{\mathbf{\Omega}}_{jl}^{(k)}(\lambda) \ne \mathbf{0} \bigg)}{\sum_{1 \leq j < l \leq p}\mathds{1} \bigg(\mathbf{\Omega}_{0,jl}^{(k)}=\mathbf{0} \bigg)}\\
	\text{ATPR}(\lambda)&=\frac{1}{K}\sum_{k=1}^{K}\frac{\sum_{1 \leq j < l \leq p}\mathds{1} \bigg (\mathbf{\Omega}_{0,jl}^{(k)}\ne\mathbf{0} \, , \,\hat{\mathbf{\Omega}}_{jl}^{(k)}(\lambda) \ne\mathbf{0} \bigg)}{\sum_{1 \leq j < l \leq p}\mathds{1} \bigg(\mathbf{\Omega}_{0,jl}^{(k)}\ne\mathbf{0} \bigg)},
\end{align*}
where $\mathds{1}$ is the indicator function, $\mathbf{\Omega}_{0}^{(k)}$ is the true precision matrix and $\hat{\mathbf{\Omega}}^{(k)}(\lambda)$ is the estimated precision matrix using tuning parameter $\lambda$. The above quantities are calculated for 100 values of $\lambda$, where 90 of them are in $[0,0.67]$ and 10 of them are in $[0.6784,1.5]$. All of them are equally spaced in their respective intervals and start and end at the boundaries of their respective intervals. Each scenario is simulated 5 times, and the final ROC curve is the average of them.

\begin{figure}
	\centering
	\begin{tabular}{ccc}
		\includegraphics[scale=0.33 ]{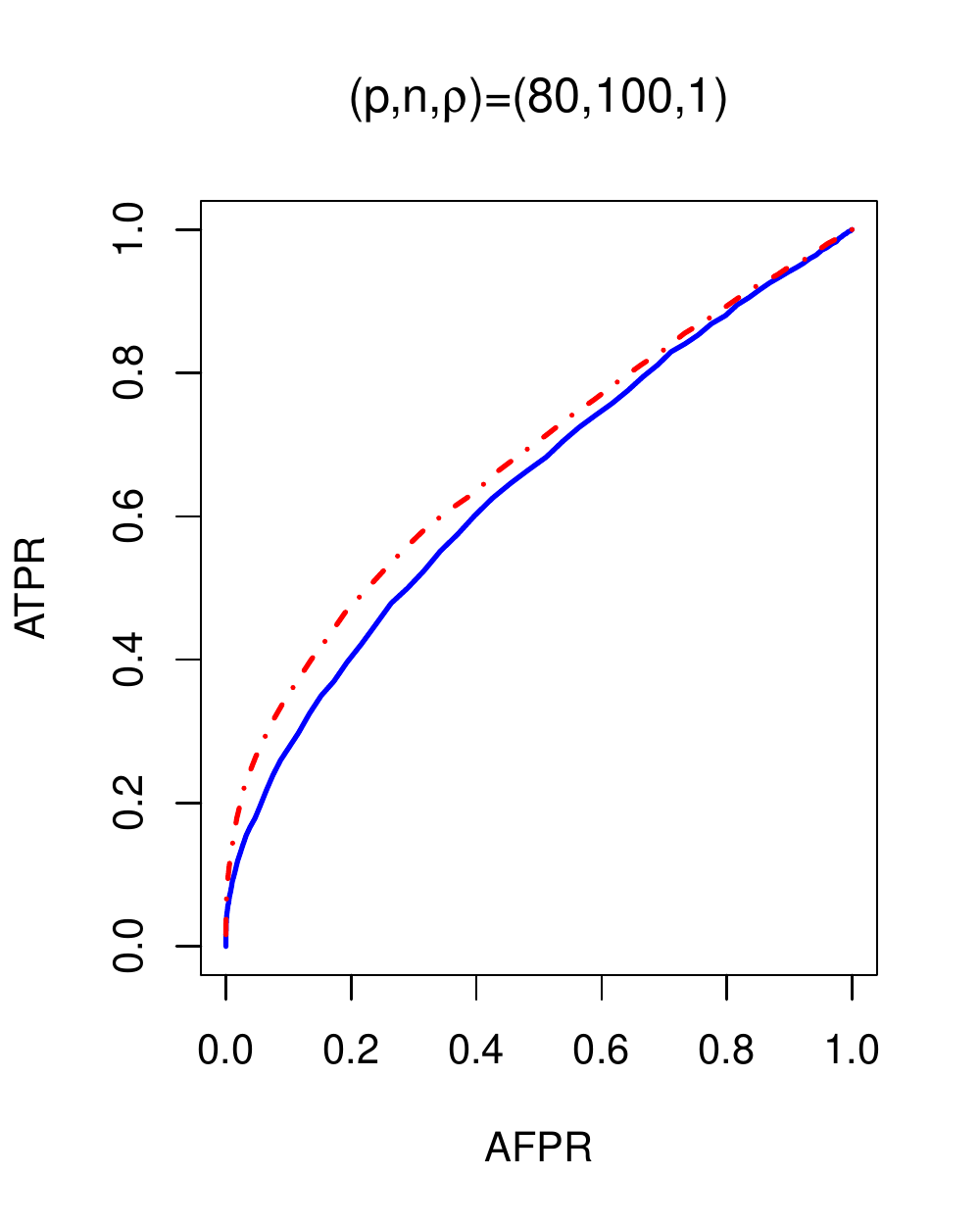} & \includegraphics[scale=0.33 ]{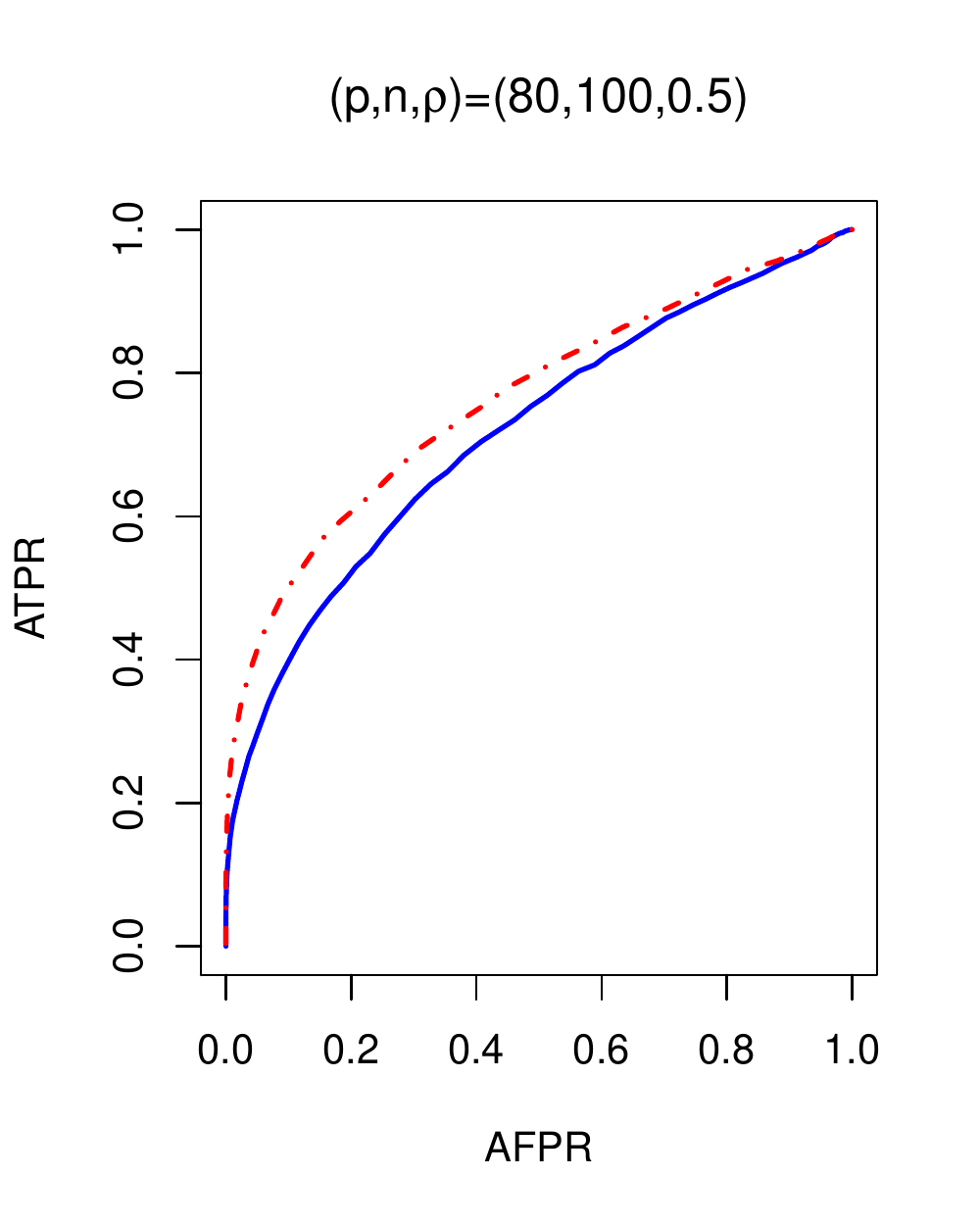} & 
		\includegraphics[scale=0.33 ]{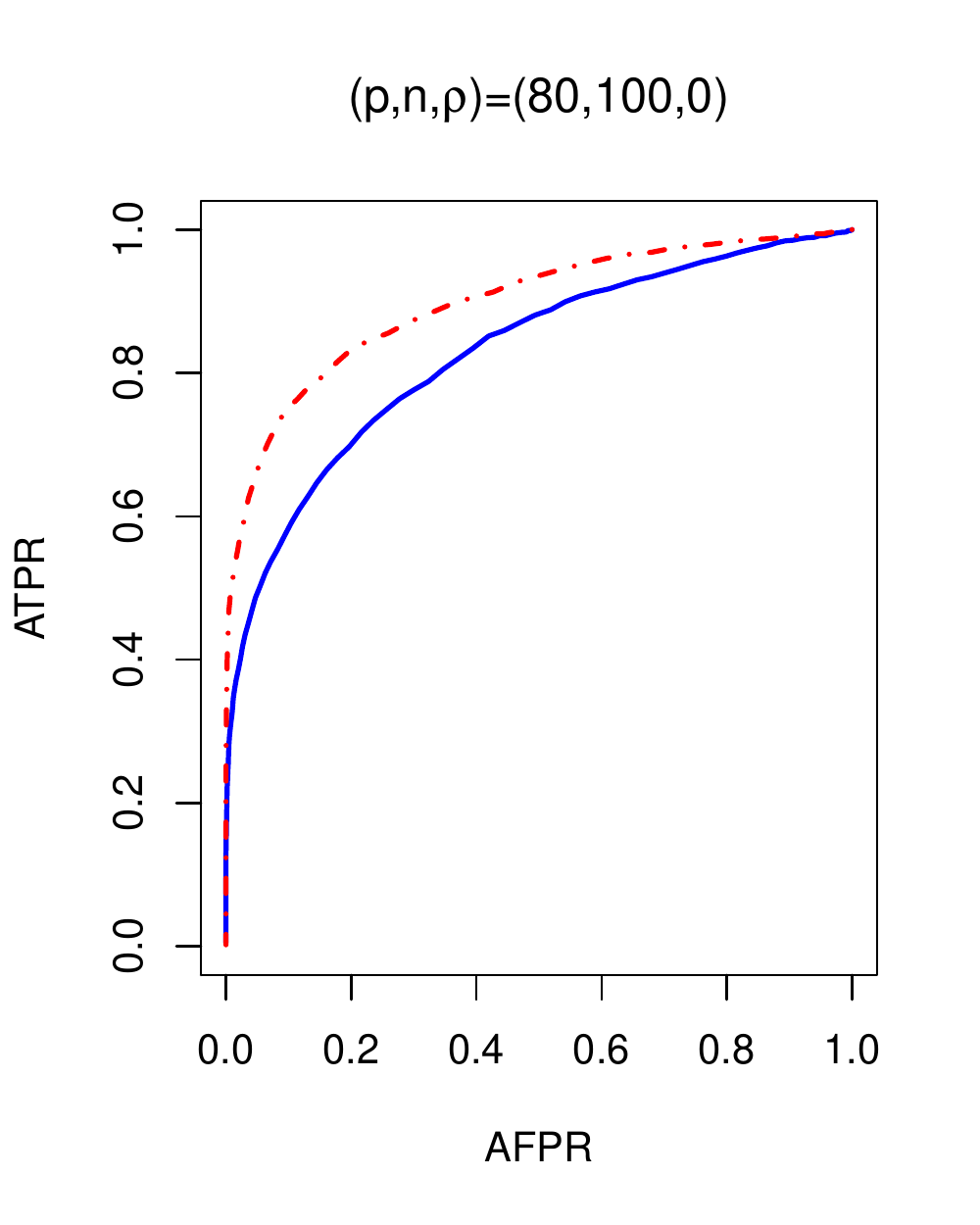}\\
		\includegraphics[scale=0.33 ]{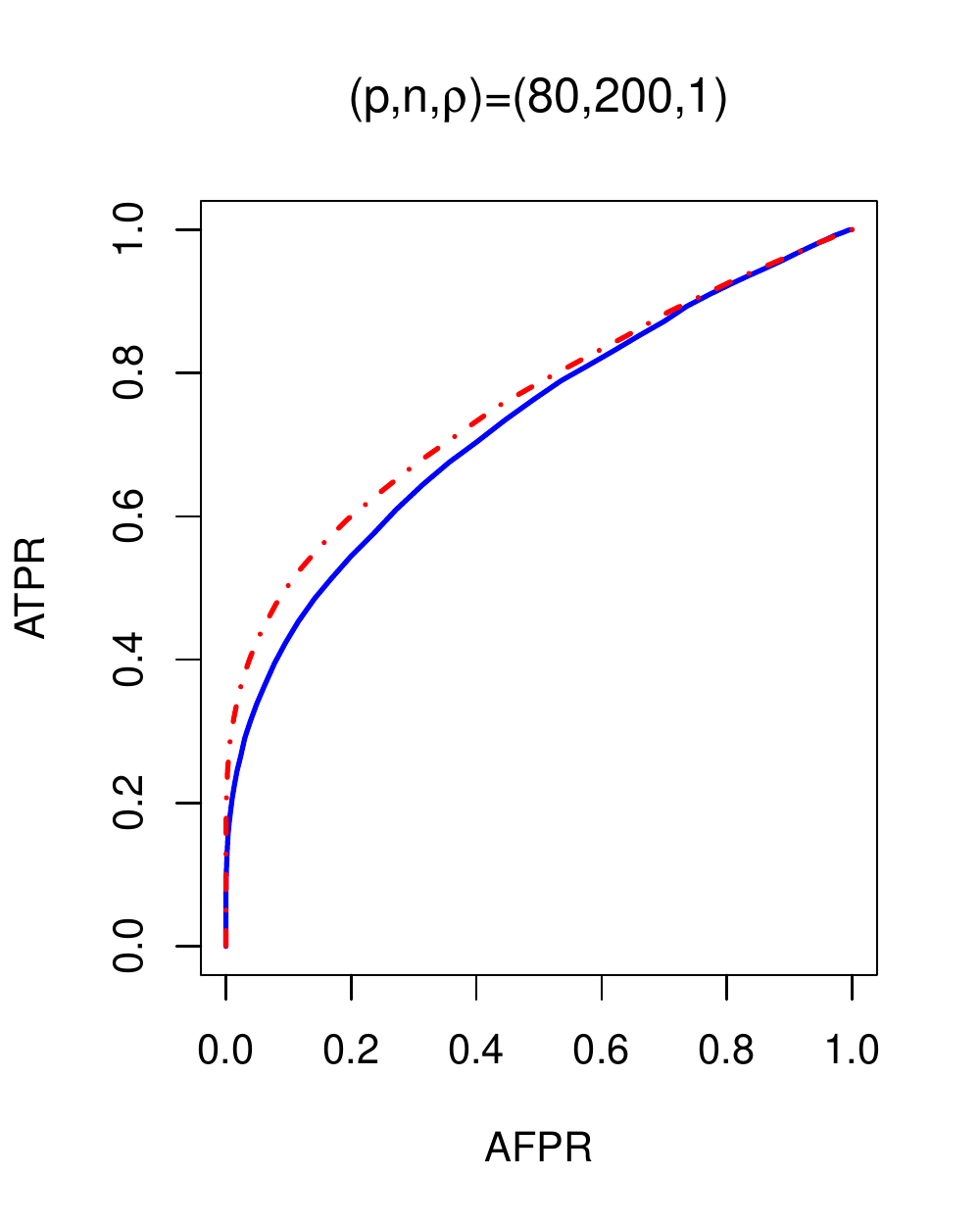} & \includegraphics[scale=0.33 ]{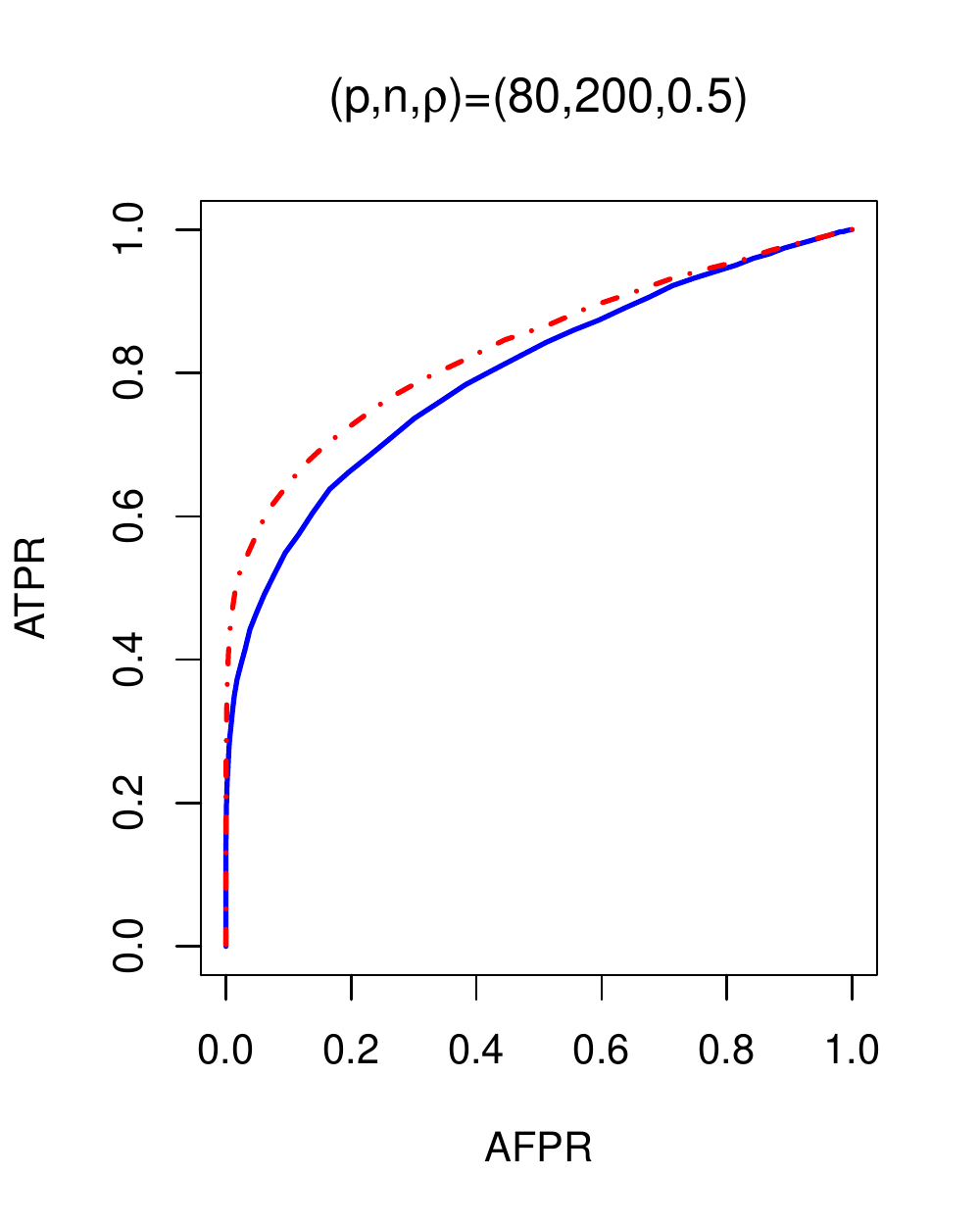} & 
		\includegraphics[scale=0.33 ]{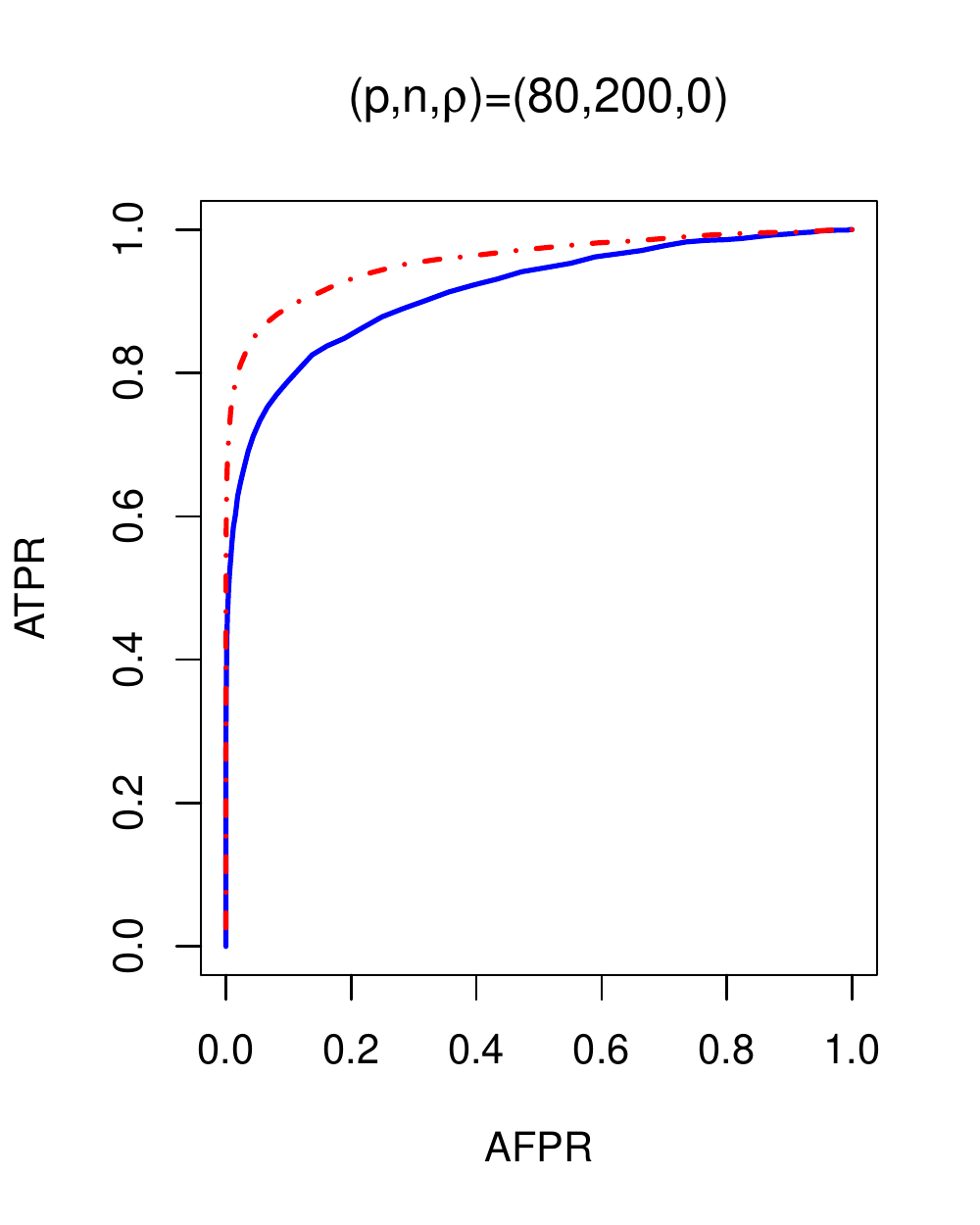}\\
		\includegraphics[scale=0.33 ]{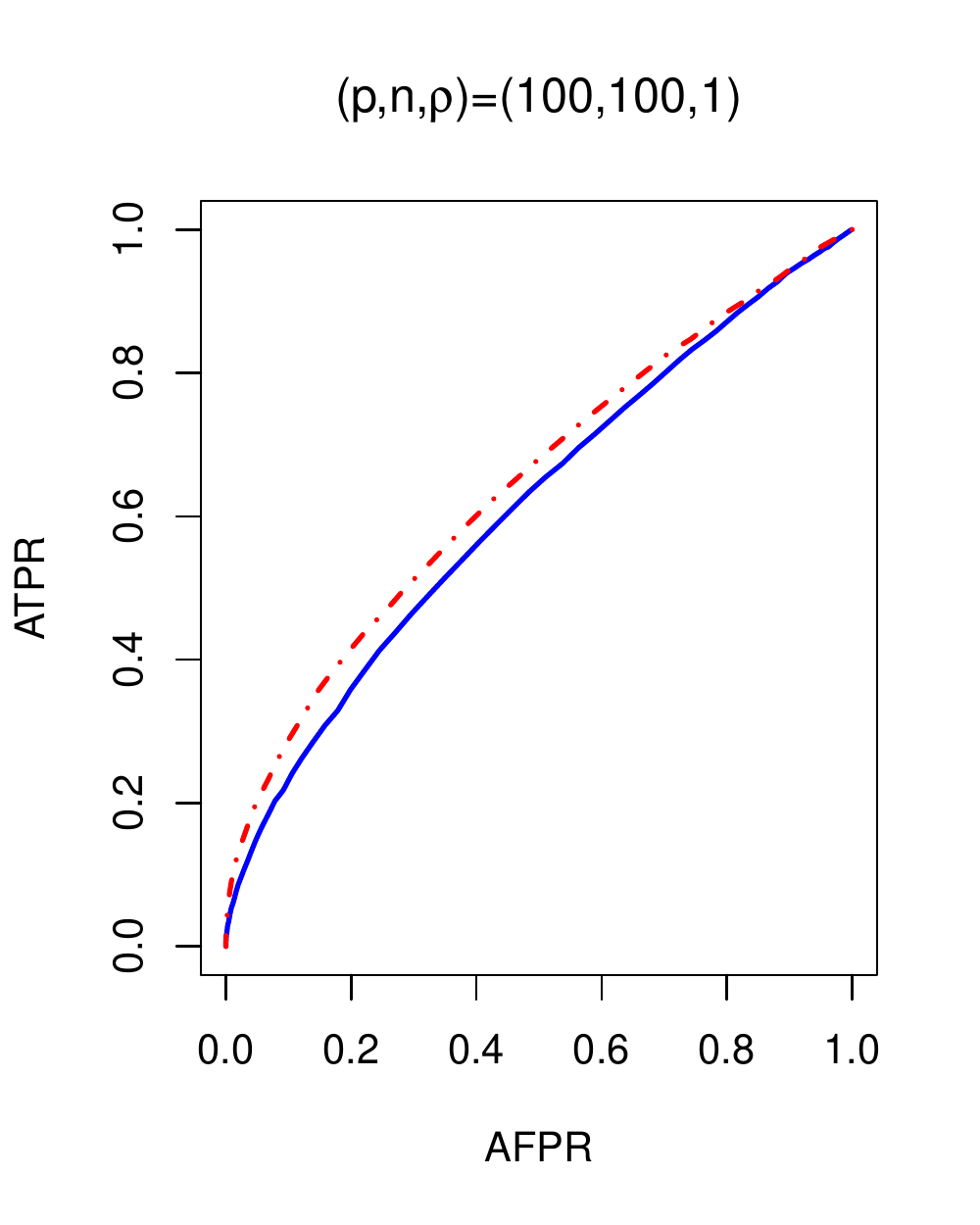} & \includegraphics[scale=0.33 ]{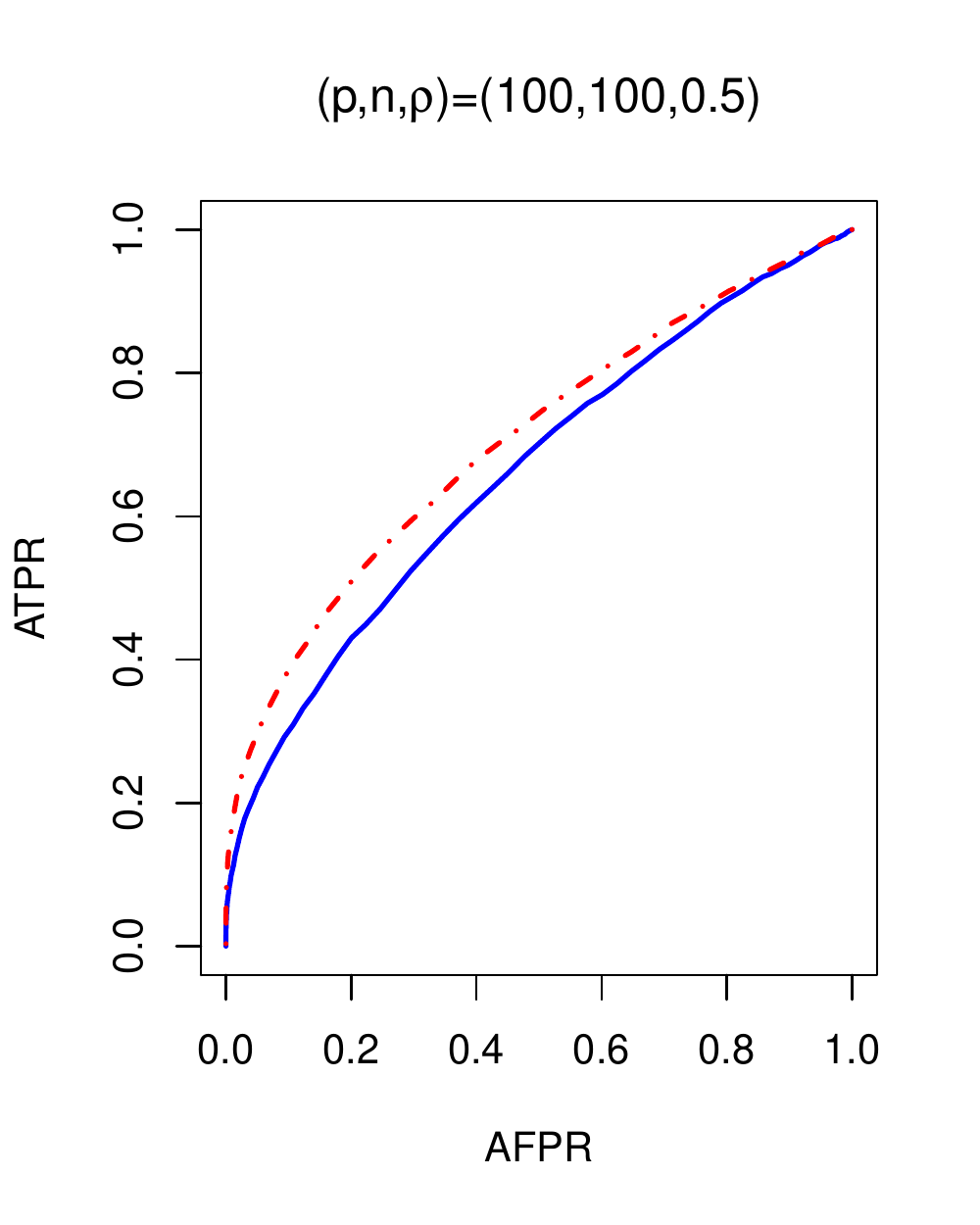} & 
		\includegraphics[scale=0.33 ]{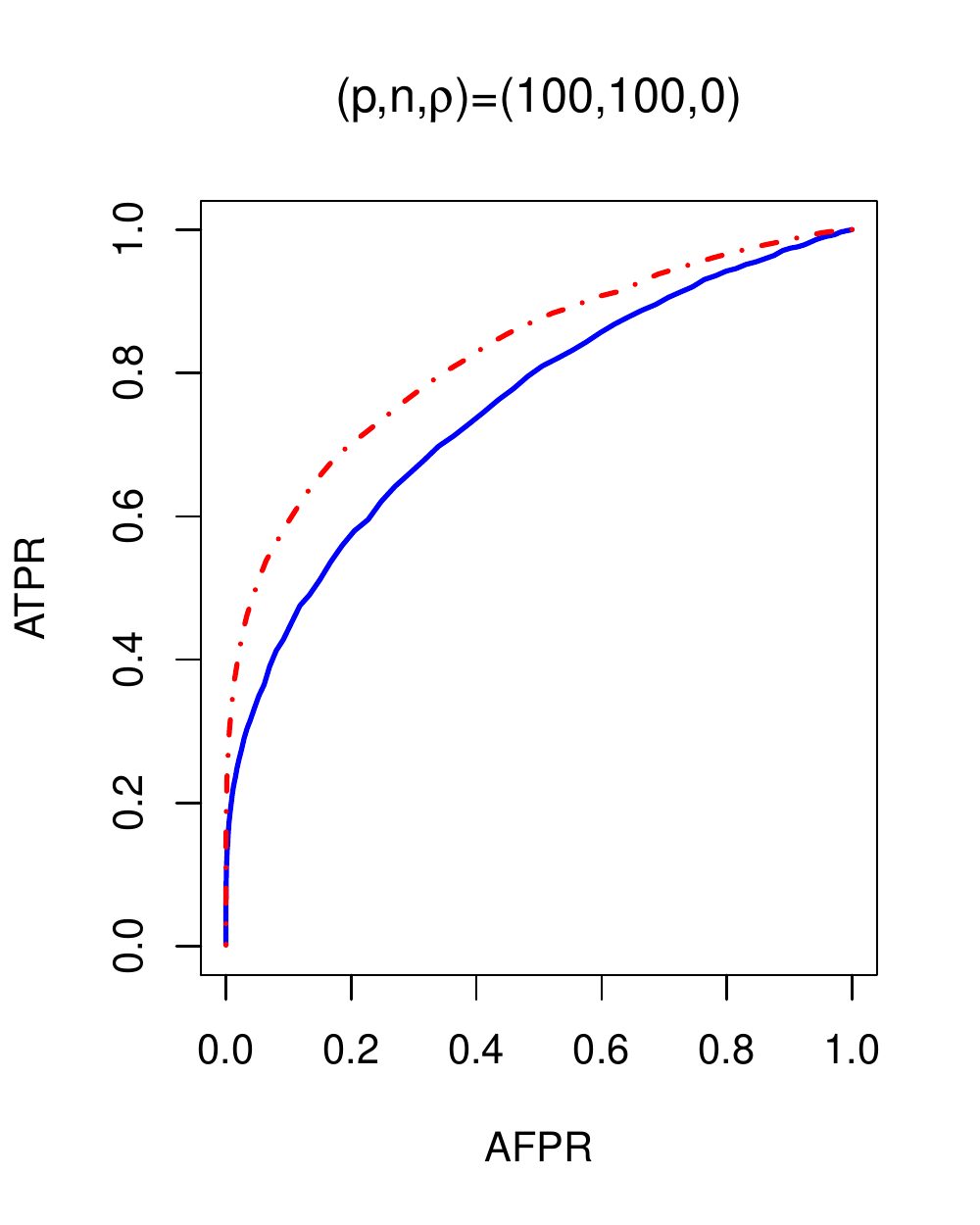}\\
		\includegraphics[scale=0.33 ]{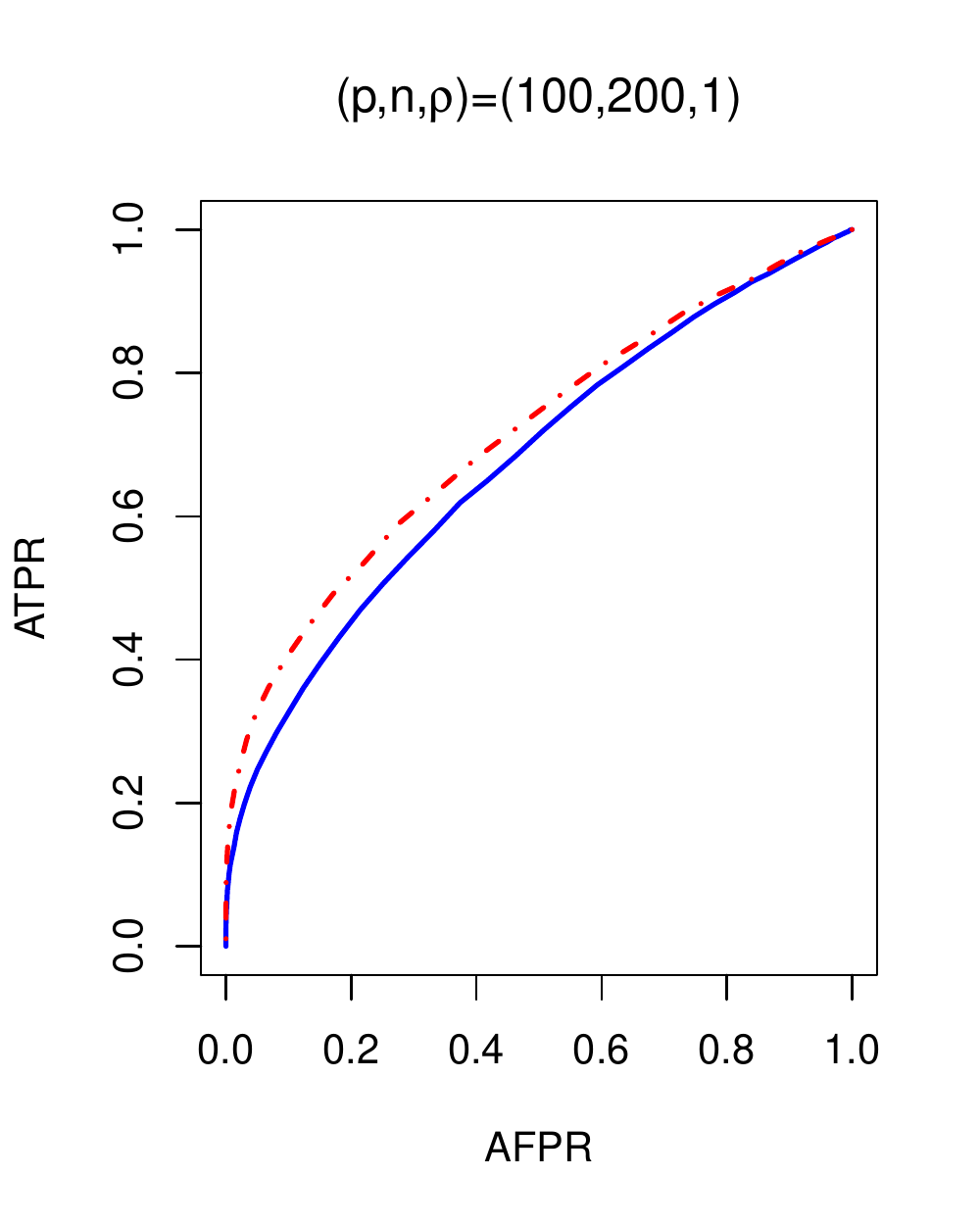} & \includegraphics[scale=0.33 ]{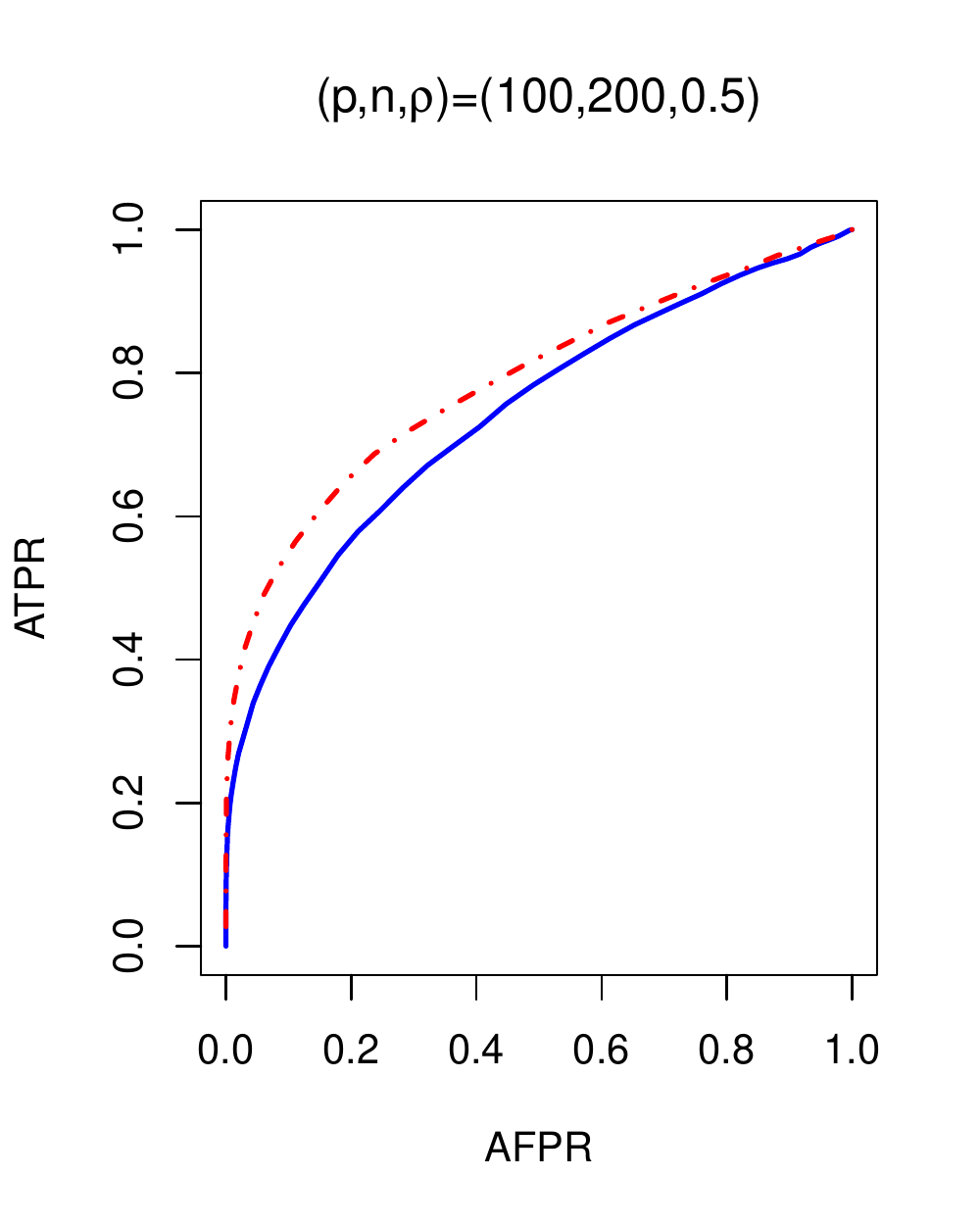} & \includegraphics[scale=0.33 ]{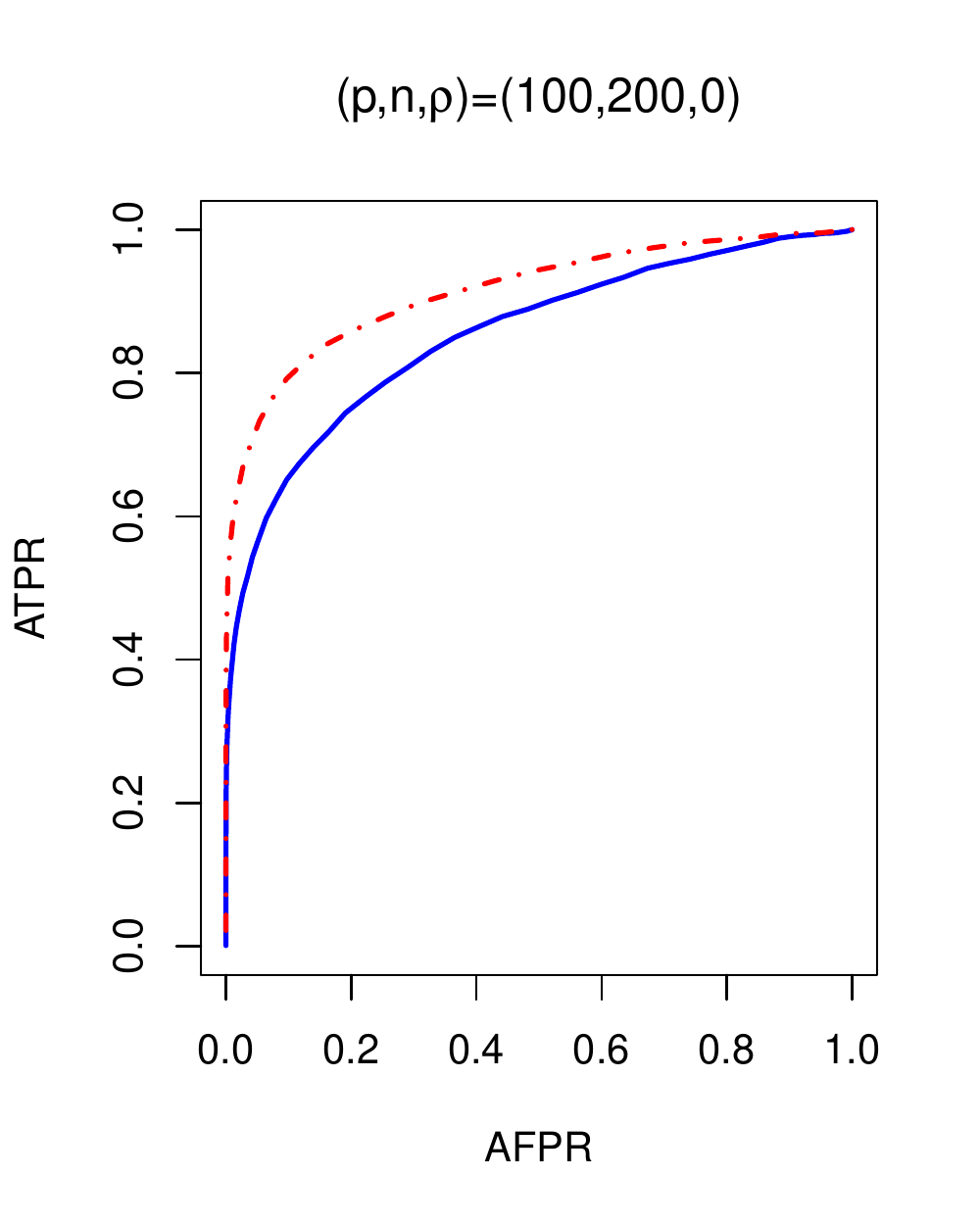}
	\end{tabular}
	\caption{ROC curves by JFGGM (red dashed line) and by separate estimation with FGGM (solid blue line) for different combinations of $(p,n,\rho)$.}
	\label{figure:1}
\end{figure}

\begin{table}
	\centering
	\begin{tabular}{|c|c|c|c|c|}
		\hline
		$\big((p,n),\rho\big)$&\text{Method}&1 &0.5&0\\
		\hline
		\multirow{2}{4em}{$(80,100)$}&\text{JFGGM}& 0.68  & 0.76 & 0.90 \\\cline{2-5}
		& \text{FGGM}& 0.64& 0.72& 0.83   \\
		\hline
		\multirow{2}{4em}{$(80,200)$}&\text{JFGGM}& 0.76& 0.83 & 0.96   \\\cline{2-5}
		& \text{FGGM}& 0.73 & 0.80& 0.91  \\
		\hline
		\multirow{2}{4em}{$(100,100)$}&\text{JFGGM}& 0.65& 0.70 & 0.83   \\\cline{2-5}
		& \text{FGGM}& 0.61& 0.66& 0.75   \\
		\hline
		\multirow{2}{4em}{$(100,200)$}&\text{JFGGM}& 0.71& 0.79 & 0.91   \\\cline{2-5}
		& \text{FGGM}& 0.68& 0.74& 0.85   \\
		\hline
	\end{tabular}
	\caption{Table with the area under the ROC curves from Figure \ref*{figure:1} }
	\label{tabel:1}
\end{table}

Figure \ref{figure:1} shows the ROC curves by the JFGGM and those by separate estimation with FGGM. Overall, our method outperforms separate estimation with the FGGM. When the number of individual edges is the same with the number of common edges the two methods are close. As the number of individual edges decreases, JFGGM significantly outperforms separate estimation. In addition, the JFGGM performs well on high dimensions with a relatively small sample size. Table \ref{tabel:1} provides the area under the curve for each scenario, verifying the visual results. In all scenarios, the ADMM algorithm produced accurate estimators after no more than 100 iterations.

\section{Application}

\begin{figure}
	\begin{center}
		\includegraphics[scale=0.6]{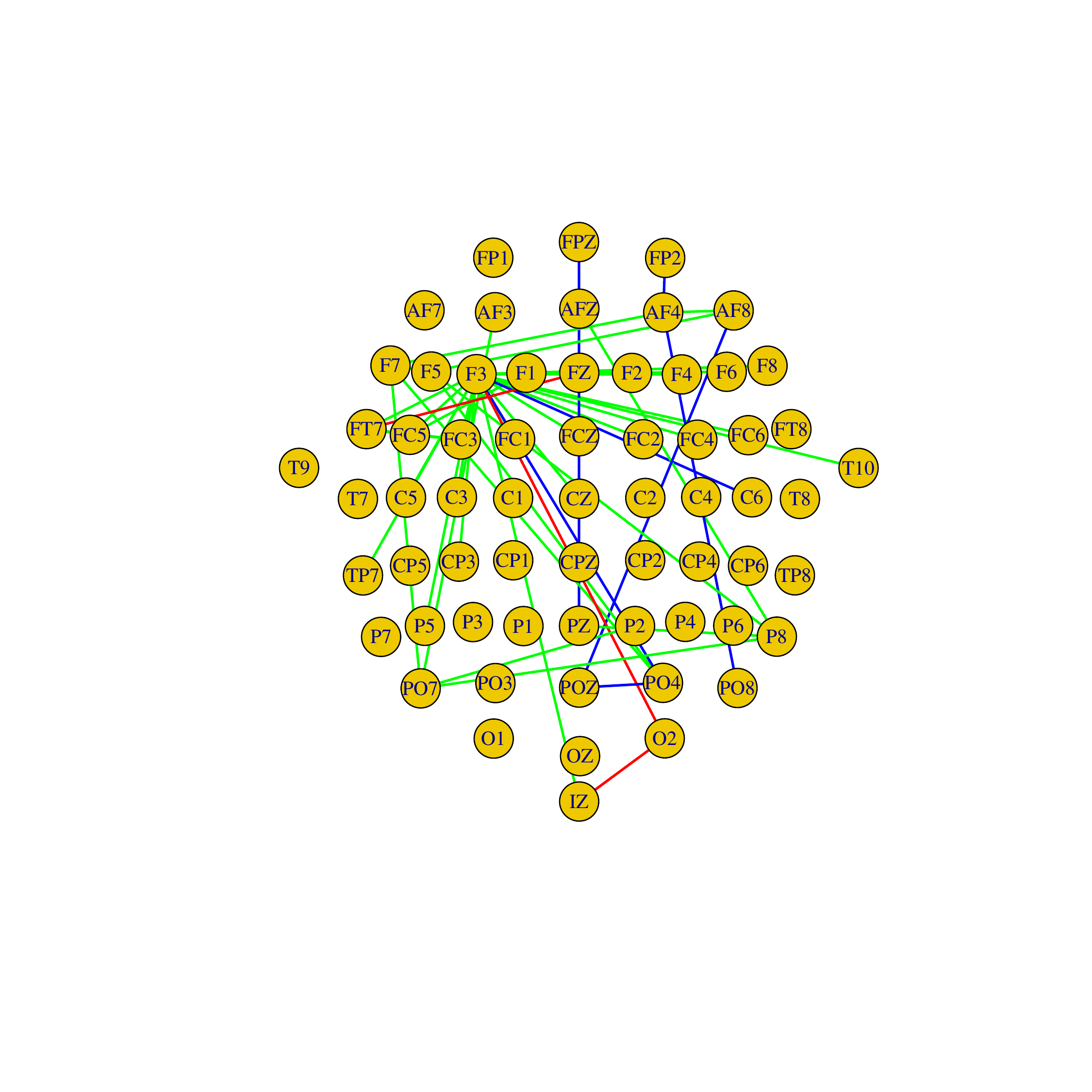}
	\end{center}
	\caption{Graph of the 64 electrodes, produced by the JFGGM. Green represents the common edges, red the edges unique to the alcoholic group, and blue the edges unique to the non-alcoholic group.}
	\label{figure:2}
\end{figure}

In this section we apply the JFGGM to the EEG dataset mentioned in the Introduction. The dataset consists of two groups of subjects: alcoholic and non-alcoholic. The first group is comprised of 77 subjects and the second of 45. Sixty four electrodes were strategically placed on each subject's scalp, which measured their brain activity while they were shown pictures of a variety of objects. Measurements of brain activity were sampled at 256 Hz for 1 second. The purpose of this study is to uncover genetic predisposition to alcoholism.

\cite{li2010dimension} applied a dimension folding method to the EEG dataset where brain activity recorded at each electrode was treated as a multivariate random vector of 256 entries. \cite{qiao2019functional} and \cite{soleali} treated the same quantities as stochastic processes. Both of them, however, apply their methods separately to the alcoholic and the control group, losing the joint information for prediction. In contrast, with the JFGGM we can exploit the information that exists across these groups with joint estimation of their graphs. It is also computationally efficient, since  we would have to choose two $\lambda$'s for the separate estimation, one for each group. Finally, JFGGM makes comparison of the two graphs easier, since the level of sparsity for both is controlled by the same $\lambda$.

Figure \ref{figure:2} shows the graph estimated by the JFGGM among the 64 stochastic processes describing the brain activity at the electrodes. The layout of the vertices represents the position of the electrodes on the scalp, with the top side being the front of the skull. We chose $\lambda=3.5$ so that the sparsity level of the graph is at $2.5\%$. For every stochastic process of each group, we estimated a Karhunen-Loeve expansion of $M=5$ eigenfunctions.

From Figure \ref{figure:2} we see that the graphs of the two groups have a rich common structure, which indicates that our joint estimation procedure has a significant advantage. The majority of the common structure is located at the front side of the left hemisphere of the scalp. Furthermore, the two groups exhibit important differences: the edges unique to the alcoholic group are observed on an acute diagonal strip near the center of the scalp, whereas the edges unique to the non-alcoholic group occupy the right hemisphere of the scalp. Finally, the non-alcoholic group has seven extra individual edges, while the alcoholic group has only three, indicating heightened brain activity in the control group.

\section{Discussion}\label{discussion}
In this paper, we develop a method for jointly estimating functional graphical models. The assumption is that these graphs share a significant common structure. We can see the common structure in the distribution of the Karhunen-Loeve expansion coefficients $\mathcal{N}(\mathbf{0},\mathbf{\Omega}^{(k)})$, for each subpopulation $k$. Each precision matrix $\mathbf{\Omega}^{(k)}$ can be decomposed into the Hadamard product of a matrix $\mathbf{\Theta}$, that is common for all subpopulations and expresses the common structure, and a subpopulation specific matrix $\mathbf{\Gamma}^{(k)}$, that expresses the individual structure of the $k$-th subpopulation. By estimating a single graph for all the data, we would be ignoring the individuality of each subpopulation. On the other hand, by estimating a single graph for each subpopulation, we would not be using the existence of the common structure to our advantage. We accommodate our method with two optimization algorithms. The first dealing with the nonconvex nature of the objective function, and the second dealing with the nonsmooth nature of the first algorithm. To complete the theoretical novelty of our model, we establish the asymptotic consistency of our estimator.
The theoretical accuracy of the JFGGM is demonstrated in a simulation experiment against a separate estimation of the graphs with the FGGM method developed in \cite{qiao2019functional}.

To conclude our work, we would like to point three possible extensions of this article. First, we have assumed that the gaussian processes associated with each subpopulation are realizations of Hilbert spaces with the same dimension $M$. Hence, the first possibility could be to extended this method to the case where the dimension of the Hilbert spaces is subpopulation specific. Second, the assumption that the Karhunen-Loeve expansion coefficients follow directly a $\mathcal{N}(\mathbf{0},\mathbf{\Omega}^{(k)})$ for each subpopulation $k$, can be further relaxed by assuming instead that the coefficients may not follow initially a multivariate normal, but there exists a transformation such that the transformed coefficients follow it, as described in \cite{soleali}. A third possible extension would be to get rid of the multivariate normal distributional assumption altogether for all subpopulations by replacing the conditional independence relationship which defines the graphs with the additive conditional independence relationship studied in \cite{li2018nonparametric}.

\appendix
\section{Proving Theorem~\ref{theorem:2}}\label{appendix:a}

For a matrix $\mathbf{A}=(a_{jl})$, let $\|\cdot\|_{1}$ denote the usual $\ell_{1}$ vector norm $\|\mathbf{A}\|_{1}=\sum_{j,l}|a_{jl}|$. Define the objective function
\begin{align}\label{obj:2}
	\sum_{k=1}^{K}\left[\tr\left(\hat{\mathbf{\Sigma}}^{(k)}\mathbf{\Omega}^{(k)}\right)-\log\det\left(\mathbf{\Omega}^{(k)}\right)\right]+\sum_{j\ne l}\theta_{jl}+\lambda_{1}\lambda_{2}\sum_{j\ne l}\sum_{k=1}^{K}\|\mathbf{\Gamma}_{jl}^{(k)}\|_{F},  
\end{align}
for $\theta_{jl}$ and $\mathbf{\Gamma}_{jl}^{(k)}$ specified above.  We first show that the objective functions \eqref{obj:1} and \eqref{obj:2} are equivalent.

\medskip
\begin{lemma}\label{lemma:1}
	Suppose $(\hat{\mathbf{\Theta}},\hat{\mathbf{\Gamma}})$ is a local minimizer of \eqref{obj:1}. Then, there exists a local minimizer $(\tilde{\mathbf{\Theta}},\tilde{\mathbf{\Gamma}})$ of \eqref{obj:2} such that $\tilde{\theta}_{jl}\tilde{\mathbf{\Gamma}}_{jl}^{(k)}=\hat{\theta}_{jl}\hat{\mathbf{\Gamma}}_{jl}^{(k)}$ for all $j,l,k$. Conversely, let $(\tilde{\mathbf{\Theta}},\tilde{\mathbf{\Gamma}})$ be a local minimizer of \eqref{obj:2}. Then, there exists a local minimizer $(\hat{\mathbf{\Theta}},\hat{\mathbf{\Gamma}})$ of \eqref{obj:1} such that $\hat{\theta}_{jl}\hat{\mathbf{\Gamma}}_{jl}^{(k)}=\tilde{\theta}_{jl}\tilde{\mathbf{\Gamma}}_{jl}^{(k)}$ for all $j,l,k$.
\end{lemma}
\begin{proof}
	Let $Q_{1}(\lambda_{1},\lambda_{2},\mathbf{\Theta},\mathbf{\Gamma})$ and $Q_{2}(\lambda_{1}\lambda_{2},\mathbf{\Theta},\mathbf{\Gamma})$ denote the objective functions \eqref{obj:1} and \eqref{obj:2}, respectively. Observe that
	\begin{align*}
		Q_{1}(\lambda_{1},\lambda_{2},\mathbf{\Theta},\mathbf{\Gamma})&=Q_{2}(\lambda_{1}\lambda_{2},\lambda_{1}\mathbf{\Theta},\lambda_{1}^{-1}\mathbf{\Gamma})\\
		Q_{2}(\lambda_{1}\lambda_{2},\mathbf{\Theta},\mathbf{\Gamma})&=Q_{1}(\lambda_{1},\lambda_{2},\lambda_{1}^{-1}\mathbf{\Theta},\lambda_{1}\mathbf{\Gamma})
	\end{align*}
	Since $(\hat{\mathbf{\Theta}},\hat{\mathbf{\Gamma}})$ is a local minimizer of $Q_{1}(\lambda_{1},\lambda_{2},\cdot,\cdot)$, there exists $\mathbf{\delta}>0$ such that for every $(\mathbf{\Theta},\mathbf{\Gamma})$
	with
	\begin{align*}
		\|\mathbf{\Theta}-\hat{\mathbf{\Theta}}\|_{1}+\sum_{k=1}^{K}\|\mathbf{\Gamma}^{(k)}-\hat{\mathbf{\Gamma}}^{(k)}\|_{1}<\delta,
	\end{align*}
	we have
	\begin{align*}
		Q_{1}(\lambda_{1},\lambda_{2},\hat{\mathbf{\Theta}},\hat{\mathbf{\Gamma}})\leq Q_{1}(\lambda_{1},\lambda_{2},\mathbf{\Theta},\mathbf{\Gamma}).
	\end{align*}
	Let $0<\delta^{*}\leq \delta\min\left(\lambda_{1},\lambda_{1}^{-1}\right)$, and define $(\tilde{\mathbf{\Theta}},\tilde{\mathbf{\Gamma}})=(\lambda_{1}\hat{\mathbf{\Theta}},\lambda_{1}^{-1}\hat{\mathbf{\Gamma}})$. Then, for any $(\mathbf{\Theta},\mathbf{\Gamma})$ satisfying
	\begin{align*}
		\|\mathbf{\Theta}-\tilde{\mathbf{\Theta}}\|_{1}+\sum_{k=1}^{K}\|\mathbf{\Gamma}^{(k)}-\tilde{\mathbf{\Gamma}}^{(k)}\|_{1}<\delta^{*},
	\end{align*}
	we have
	\begin{align*}
		\|\lambda_{1}^{-1}\mathbf{\Theta}-\hat{\mathbf{\Theta}}\|_{1}+\sum_{k=1}^{K}\|\lambda_{1}\mathbf{\Gamma}^{(k)}-\hat{\mathbf{\Gamma}}^{(k)}\|_{1}
		\leq\frac{\|\mathbf{\Theta}-\tilde{\mathbf{\Theta}}\|_{1}+\sum_{k=1}^{K}\|\mathbf{\Gamma}^{(k)}-\tilde{\mathbf{\Gamma}}^{(k)}\|_{1}}{\min\left(\lambda_{1},\lambda_{1}^{-1}\right)}\leq\delta.
	\end{align*}
	Thus
	\begin{align*}
		Q_{1}(\lambda_{1},\lambda_{2},\hat{\mathbf{\Theta}},\hat{\mathbf{\Gamma}})\leq Q_{1}(\lambda_{1},\lambda_{2},\lambda_{1}^{-1}\mathbf{\Theta},\lambda_{1}\mathbf{\Gamma})\Rightarrow Q_{2}(\lambda_{1}\lambda_{2},\tilde{\mathbf{\Theta}},\tilde{\mathbf{\Gamma}})\leq Q_{2}(\lambda_{1}\lambda_{2},\mathbf{\Theta},\mathbf{\Gamma}),
	\end{align*}
	which means that $(\tilde{\mathbf{\Theta}},\tilde{\mathbf{\Gamma}})$ is a local minimizer of \eqref{obj:2}. The other direction is proven similarly.
\end{proof}

\medskip
\begin{lemma}\label{lemma:2}
	Suppose $(\hat{\mathbf{\Theta}},\hat{\mathbf{\Gamma}})$ is a local minimizer of \eqref{obj:2} and $\hat{\mathbf{\Omega}}^{(k)}_{jl}=\hat{\theta}_{jl}\hat{\mathbf{\Gamma}}_{jl}^{(k)}$ for all $j,l,k$. Then, for $j,l\in\mathcal{V}$, $j\ne l$, the following are true:
	\begin{enumerate}
		\item $\hat{\theta}_{jl}=0$ if and only if $\hat{\mathbf{\Gamma}}_{jl}^{(k)}=0$ for all $k=1,\ldots,K$.
		\item If $\hat{\mathbf{\theta}}_{jl}\ne0$, then $\hat{\theta}_{jl}=\left(\lambda\sum_{k=1}^{K}\|\hat{\mathbf{\Omega}}_{jl}^{(k)}\|_{F}\right)^{1/2}$, where $\lambda=\lambda_{1}\lambda_{2}$.
	\end{enumerate}
	
\end{lemma}
\begin{proof}
	1. If $\theta_{jl}$ is 0, then $\mathbf{\Gamma}_{jl}^{(1)},\ldots,\mathbf{\Gamma}_{jl}^{(K)}$ only appear in the third term in \eqref{obj:2}. Thus, in order to minimize $Q_{2}$, we need $\mathbf{\Gamma}_{jl}^{(k)}=\mathbf{0}$, for all $k=1,\ldots,K$. The other direction is similar.
	
	2. Suppose $\hat{\theta}_{jl}\ne 0$ and let
	\begin{align*}
		c=\frac{\left(\lambda\sum_{k=1}^{K}\|\hat{\mathbf{\Omega}}_{jl}^{(k)}\|_{F}\right)^{1/2}}{\hat{\theta}_{jl}}.
	\end{align*}
	We will show $c=1$. By definition,
	\begin{align*}
		\hat{\mathbf{\Gamma}}_{jl}^{(k)}=c\frac{\hat{\mathbf{\Omega}}_{jl}^{(k)}}{\left(\lambda\sum_{k=1}^{K}\|\hat{\mathbf{\Omega}}_{jl}^{(k)}\|_{F}\right)^{1/2}}.
	\end{align*}
	Suppose $c>1$. Since $(\hat{\mathbf{\Theta}},\hat{\mathbf{\Gamma}})$ is a local minimizer of $Q_{2}(\lambda,\cdot,\cdot)$, there exists $\delta>0$ such that for all $(\mathbf{\Theta},\mathbf{\Gamma})$ with
	\begin{align*}
		\|\mathbf{\Theta}-\hat{\mathbf{\Theta}}\|_{1}+\sum_{k=1}^{K}\|\mathbf{\Gamma}^{(k)}-\hat{\mathbf{\Gamma}}^{(k)}\|_{1}<\delta
	\end{align*}
	we have $Q_{2}(\lambda,\hat{\mathbf{\Theta}},\hat{\mathbf{\Gamma}})\leq Q_{2}(\lambda,\mathbf{\Theta},\mathbf{\Gamma})$.
	
	Then there exists $\delta^{*}\in(1,c)$, slightly greater than 1, such that for $(\tilde{\mathbf{\Theta}},\tilde{\mathbf{\Gamma}})$ defined by
	\begin{align*}
		\begin{cases}
			\tilde{\theta}_{j'l'}=\hat{\theta}_{j'l'}\text{ and }\tilde{\mathbf{\Gamma}}_{j'l'}^{(k)}=\hat{\mathbf{\Gamma}}_{j'l'}^{(k)}, &(j',l')\ne(j,l)\\
			\tilde{\theta}_{jl}=\delta^{*}\hat{\theta}_{jl}\text{ and }\tilde{\mathbf{\Gamma}}_{jl}^{(k)}=\frac{1}{\delta^{*}}\hat{\mathbf{\Gamma}}_{jl}^{(k)}
		\end{cases}
	\end{align*}
	we have
	\begin{align*}
		\|\tilde{\mathbf{\Theta}}-\hat{\mathbf{\Theta}}\|_{1}+\sum_{k=1}^{K}\|\tilde{\mathbf{\Gamma}}^{(k)}-\hat{\mathbf{\Gamma}}^{(k)}\|_{1}<\delta.
	\end{align*}
	But this implies
	\begin{align*}
		&Q_{2}(\lambda,\hat{\mathbf{\Theta}},\hat{\mathbf{\Gamma}})-Q_{2}(\lambda,\tilde{\mathbf{\Theta}},\tilde{\mathbf{\Gamma}})\\
		=&(1-\delta^{*})\hat{\theta}_{jl}+\left(1-\frac{1}{\delta^{*}}\right)\lambda\sum_{k=1}^{K}\|\hat{\mathbf{\Gamma}}_{jl}^{(k)}\|_{F}\\
		=&\frac{1}{c}(\delta^{*}-1)\left(\frac{c^{2}}{\delta^{*}}-1\right)\left(\lambda\sum_{k=1}^{K}\|\hat{\mathbf{\Omega}}_{jl}^{(k)}\|_{F}\right)^{1/2}>0,
	\end{align*}
	which is impossible because $(\hat{\mathbf{\Theta}},\hat{\mathbf{\Gamma}})$ is a local minimizer. Hence $c\leq 1$. Following the same argument we can show $c\geq 1$. Thus $c=1$.
\end{proof}

We are now ready to establish the equivalence between the objective functions \eqref{obj:1} and \eqref{obj:3}, for which it suffices to show the equivalence of \eqref{obj:3} and \eqref{obj:2}.

\medskip
\begin{lemma}\label{lemma:4}
	Let $(\hat{\mathbf{\Theta}},\hat{\mathbf{\Gamma}})$ be a local minimizer of \eqref{obj:2}. Then, there exists a local minimizer $\hat{\mathbf{\Omega}}$ of \eqref{obj:3} such that $\hat{\mathbf{\Omega}}_{jl}^{(k)}=\hat{\theta}_{jl}\hat{\mathbf{\Gamma}}_{jl}^{(k)}$ for all $j,l,k$. Conversely, let $\hat{\mathbf{\Omega}}$ be a local minimizer of \eqref{obj:3}. Then, there exists a local minimizer $(\hat{\mathbf{\Theta}},\hat{\mathbf{\Gamma}})$ of \eqref{obj:2} such that $\hat{\theta}_{jl}\hat{\mathbf{\Gamma}}_{jl}^{(k)}=\hat{\mathbf{\Omega}}_{jl}^{(k)}$ for all $j,l,k$. 
\end{lemma}
\begin{proof}
	Let $Q_{3}(\lambda_{1}\lambda_{2},\mathbf{\Omega})$ denote the objective function \eqref{obj:3}. Suppose $(\hat{\mathbf{\Theta}},\hat{\mathbf{\Gamma}})$ is a local minimizer of \eqref{obj:2}. Then, there exists $\delta>0$ such that for all $(\mathbf{\Theta},\mathbf{\Gamma})$ with
	\begin{align*}
		\|\mathbf{\Theta}-\hat{\mathbf{\Theta}}\|_{1}+\sum_{k=1}^{K}\|\mathbf{\Gamma}^{(k)}-\hat{\mathbf{\Gamma}}^{(k)}\|_{1}<\delta,
	\end{align*}
	we have
	\begin{align*}
		Q_{2}(\lambda,\hat{\mathbf{\Theta}},\hat{\mathbf{\Gamma}})\leq Q_{2}(\lambda,\mathbf{\Theta},\mathbf{\Gamma}).
	\end{align*}
	Let $\hat{\mathbf{\Omega}}$ be the estimator associated with $(\hat{\mathbf{\Theta}},\hat{\mathbf{\Gamma}})$, that is $\hat{\mathbf{\Omega}}_{jl}^{(k)}=\hat{\theta}_{jl}\hat{\mathbf{\Gamma}}_{jl}^{(k)}$ for all $j,l,k$. In order to find a neighborhood where $\hat{\mathbf{\Omega}}$ is a minimizer we need to define the constants which will appear in the course of the proof. Let
	\begin{align*}
		a&=\min\left\{\sum_{k=1}^{K}\|\hat{\mathbf{\Omega}}_{jl}^{(k)}\|_{F} : j,l\in\mathcal{V},\hat{\theta}_{jl}\ne0 \right\},\\
		b&=\max\left\{\sum_{k=1}^{K}\|\hat{\mathbf{\Omega}}_{jl}^{(k)}\|_{F} : j,l\in\mathcal{V} \right\},
	\end{align*}
	and
	\begin{align*}
		c=2p^{2}\max\left\{\lambda^{1/2},\frac{M}{\lambda^{1/2}},\left(\frac{\lambda}{2a}\right)^{1/2},\left(\frac{M^{2}}{2a}\right)^{1/2}+\left[\left(\frac{2M^{2}}{a\lambda}\right)^{1/2}+\left(\frac{b M^{2}}{a^{2}}\right)^{1/2}\right]\right\}.
	\end{align*}
	Let $0<\delta^{*}<\min\left(\frac{a}{2},1,c^{-2}\delta^{2}\right)$, and $\mathbf{\Delta}=(\mathbf{\Delta}^{(1)},\ldots,\mathbf{\Delta}^{(K)})$, where $\mathbf{\Delta}^{(k)}=(\mathbf{\Delta}_{jl}^{(k)})\in\mathbb{R}^{pM\times pM}$, $\mathbf{\Delta}_{jl}^{(k)}\in\mathbb{R}^{M\times M}$, that satisfies
	\begin{align*}
		0<\|\mathbf{\Delta}_{jl}^{(k)}\|_{F}<\min\left\{\|\hat{\mathbf{\Omega}}_{jl}^{(k)}\|_{F}:j,l\in\mathcal{V},\hat{\theta}_{jl}\ne 0\right\},
	\end{align*}
	for all $j,l,k$, and $\sum_{k=1}^{K}\|\mathbf{\Delta}^{(k)}\|_{F}<\delta^{*}$. Let $\tilde{\mathbf{\Omega}}=\hat{\mathbf{\Omega}}+\mathbf{\Delta}$. Then
	\begin{align*}
		\sum_{k=1}^{K}\|\tilde{\mathbf{\Omega}}^{(k)}-\hat{\mathbf{\Omega}}^{(k)}\|_{1}<\delta^{*}.
	\end{align*}
	This means that $\tilde{\mathbf{\Omega}}$ is a generic element of the ball with radius less than $\delta^{*}$ and center $\hat{\mathbf{\Omega}}$.
	
	Define $(\tilde{\mathbf{\Theta}},\tilde{\mathbf{\Gamma}})$ by 
	\begin{align*}
		\tilde{\theta}_{jl}=\left(\lambda\sum_{k=1}^{K}\|\hat{\mathbf{\Omega}}_{jl}^{(k)}+\mathbf{\Delta}_{jl}^{(k)}\|_{F}\right)^{1/2}\text{ and }\tilde{\mathbf{\Gamma}}_{jl}^{(k)}=\frac{\hat{\mathbf{\Omega}}_{jl}^{(k)}+\mathbf{\Delta}_{jl}^{(k)}}{\left(\lambda\sum_{k=1}^{K}\|\hat{\mathbf{\Omega}}_{jl}^{(k)}+\mathbf{\Delta}_{jl}^{(k)}\|_{F}\right)^{1/2}},
	\end{align*}
	for all $j,l,k$. By Lemma \ref{lemma:2}, 
	\begin{align*}
		Q_{2}(\lambda,\hat{\mathbf{\Theta}},\hat{\mathbf{\Gamma}})=Q_{3}(\lambda,\hat{\mathbf{\Omega}}),
	\end{align*}
	and by the definition of $(\tilde{\mathbf{\Theta}},\tilde{\mathbf{\Gamma}})$,
	\begin{align*}
		Q_{2}(\lambda,\tilde{\mathbf{\Theta}},\tilde{\mathbf{\Gamma}})=Q_{3}(\lambda,\tilde{\mathbf{\Omega}}).
	\end{align*}
	
	If $\hat{\theta}_{jl}=0$, then $\hat{\mathbf{\Gamma}}_{jl}^{(k)}=\hat{\mathbf{\Omega}}_{jl}^{(k)}=\mathbf{0}$ for all k. Hence,
	\begin{align*}
		|\tilde{\theta}_{jl}-\hat{\theta}_{jl}|= \left(\lambda\sum_{k=1}^{K}\|\mathbf{\Delta}_{jl}^{(k)}\|_{F}\right)^{1/2}\leq \left(\lambda\sum_{k=1}^{K}\|\mathbf{\Delta}^{(k)}\|_{F}\right)^{1/2}<\lambda^{1/2}\delta^{*1/2},
	\end{align*}
	and
	\begin{align*}
		\sum_{k=1}^{K}\|\tilde{\mathbf{\Gamma}}_{jl}^{(k)}-\hat{\mathbf{\Gamma}}_{jl}^{(k)}\|_{1}&=\frac{\sum_{k=1}^{K}\|\mathbf{\Delta}_{jl}^{(k)}\|_{1}}{\left(\lambda\sum_{k=1}^{K}\|\mathbf{\Delta}_{jl}^{(k)}\|_{F}\right)^{1/2}}\\
		&\leq \frac{M}{\lambda^{1/2}}\frac{\sum_{k=1}^{K}\|\mathbf{\Delta}_{jl}^{(k)}\|_{F}}{\left(\sum_{k=1}^{K}\|\mathbf{\Delta}_{jl}^{(k)}\|_{F}\right)^{1/2}}\\
		&\leq\frac{M}{\lambda^{1/2}}\delta^{*1/2}.	
	\end{align*}
	If $\hat{\theta}_{jl}\ne 0$, then
	\begin{align*}
		|\tilde{\theta}_{jl}-\hat{\theta}_{jl}|&\leq\frac{\lambda^{1/2}}{2}\frac{\sum_{k=1}^{K}\|\mathbf{\Delta}_{jl}^{(k)}\|_{F}}{\left(\sum_{k=1}^{K}\|\hat{\mathbf{\Omega}}_{jl}^{(k)}\|_{F}-\sum_{k=1}^{K}\|\mathbf{\Delta}_{jl}^{(k)}\|_{F}\right)^{1/2}}\\
		&\leq\frac{\lambda^{1/2}}{2}\frac{\delta^{*1/2}}{\sqrt{a-\frac{a}{2}}}\\
		&=\left(\frac{\lambda}{2a}\right)^{1/2}\delta^{*1/2},
	\end{align*}
	and 
	\begin{align*}
		\sum_{k=1}^{K}\|\tilde{\mathbf{\Gamma}}_{jl}^{(k)}-\hat{\mathbf{\Gamma}}_{jl}^{(k)}\|_{1}\leq I_{1}+I_{2},
	\end{align*}
	where
	\begin{align*}
		I_{1}&=\lambda^{-1/2}\frac{\sum_{k=1}^{K}\|\mathbf{\Delta}_{jl}^{(k)}\|_{1}}{\left(\sum_{k=1}^{K}\|\hat{\mathbf{\Omega}}_{jl}^{(k)}+\mathbf{\Delta}_{jl}^{(k)}\|_{F}\right)^{1/2}}\\
		&\leq\frac{M}{\lambda^{1/2}}\frac{\sum_{k=1}^{K}\|\mathbf{\Delta}_{jl}^{(k)}\|_{F}}{\left(\sum_{k=1}^{K}\|\hat{\mathbf{\Omega}}_{jl}^{(k)}\|_{F}-\sum_{k=1}^{K}\|\mathbf{\Delta}_{jl}^{(k)}\|_{F}\right)^{1/2}}\\
		&\leq \left(\frac{M^{2}}{2a}\right)^{1/2}\delta^{*1/2},
	\end{align*}
	and
	\begin{align*}
		I_{2}&=\sum_{k=1}^{K}\|\hat{\mathbf{\Omega}}_{jl}^{(k)}\|_{1}\frac{\left|\left(\sum_{k=1}^{K}\|\hat{\mathbf{\Omega}}_{jl}^{(k)}+\mathbf{\Delta}_{jl}^{(k)}\|_{F}\right)^{1/2}-\left(\sum_{k=1}^{K}\|\hat{\mathbf{\Omega}}_{jl}^{(k)}\|_{F}\right)^{1/2}\right|}{\left(\sum_{k=1}^{K}\|\hat{\mathbf{\Omega}}_{jl}^{(k)}\|_{F}\right)^{1/2}\left(\sum_{k=1}^{K}\|\hat{\mathbf{\Omega}}_{jl}^{(k)}+\mathbf{\Delta}_{jl}^{(k)}\|_{F}\right)^{1/2}}\\
		&\leq M\left(\sum_{k=1}^{K}\|\hat{\mathbf{\Omega}}_{jl}^{(k)}\|_{F}\right)^{1/2}\frac{|\tilde{\theta}_{jl}-\hat{\theta}_{jl}|}{\left(\sum_{k=1}^{K}\|\hat{\mathbf{\Omega}}_{jl}^{(k)}\|_{F}-\sum_{k=1}^{K}\|\hat{\mathbf{\Delta}}_{jl}^{(k)}\|_{F}\right)^{1/2}}\\
		&\leq \left[\left(\frac{2M^{2}}{a\lambda}\right)^{1/2}+\left(\frac{b M^{2}}{a^{2}}\right)^{1/2}\right]\delta^{*1/2}.
	\end{align*}
	Therefore,
	\begin{align*}
		&\|\tilde{\mathbf{\Theta}}-\hat{\mathbf{\Theta}}\|_{1}+\sum_{k=1}^{K}\|\tilde{\mathbf{\Gamma}}^{(k)}-\hat{\mathbf{\Gamma}}\|_{1}\\
		&< p^{2}\max\left\{\lambda^{1/2},\left(\frac{\lambda}{2a}\right)^{1/2}\right\}\delta^{*1/2}\\
		&+p^{2}\max\left\{\frac{M}{\lambda^{1/2}},\left(\frac{M^{2}}{2a}\right)^{1/2}+\left[\left(\frac{2M^{2}}{a\lambda}\right)^{1/2}+\left(\frac{b M^{2}}{a^{2}}\right)^{1/2}\right]\right\}\delta^{*1/2}\\
		&\leq c\delta^{*1/2}\\
		&<\delta.
	\end{align*}
	Thus
	\begin{align*}
		Q_{2}(\lambda,\hat{\mathbf{\Theta}},\hat{\mathbf{\Gamma}})\leq Q_{2}(\lambda,\tilde{\mathbf{\Theta}},\tilde{\mathbf{\Gamma}})\Rightarrow Q_{3}(\lambda,\hat{\mathbf{\Omega}})\leq Q_{3}(\lambda,\tilde{\mathbf{\Omega}}),
	\end{align*}
	which means that $\hat{\mathbf{\Omega}}$ is a local minimizer of \eqref{obj:3}. The other direction is proven similarly.
\end{proof}

Combining the results from Lemmas \ref{lemma:2} and \ref{lemma:4} yields Theorem~\ref{theorem:2}.

\section{Important norm inequalities}\label{appendix:b}

\begin{lemma}
	For $1\leq j,l\leq p$, let $\mathbf{A}=(\mathbf{A}_{jl})$ and  $\mathbf{B}=(\mathbf{B}_{jl})$ be block matrices with $\mathbf{A}_{jl},\mathbf{B}_{jl}\in\mathbb{R}^{M\times M}$, let $\mathbf{u}=(\mathbf{u}_{j})$ be the block vector with $\mathbf{u}_{j}\in\mathbb{R}^{M}$, and let $\mathbf{x}=(\mathbf{x}_{j})$ and $\mathbf{y}=(\mathbf{y}_{j})$ be block matrices with $j$-th blocks $\mathbf{x}_{j},\mathbf{y}_{j}\in\mathbb{R}^{M\times M}$. The following norm properties hold:
	\begin{subequations}
		\begin{align}
			\|\mathbf{A}\|_{\max}^{(M)}&=\|\vect(\mathbf{A})\|_{\max}^{(M^{2})}\label{norm:1}\\
			\|\mathbf{A}\mathbf{u}\|_{\max}^{(M)}&\leq \|\mathbf{A}\|_{\infty}^{(M)}\,\|\mathbf{u}\|_{\max}^{(M)}\label{norm:2}\\
			\|\mathbf{x}^{\intercal}\mathbf{y}\|_{F}&\leq \|\mathbf{x}\|_{\max}^{(M)}\,\|\mathbf{y}\|_{1}^{(M)}\label{norm:3}\\
			\|\mathbf{A}\mathbf{x}\|_{\max}^{(M)}&\leq\|\mathbf{A}\|_{\max}^{(M)}\,\|\mathbf{x}\|_{1}^{(M)}\label{norm:4}\\
			\|\mathbf{A}\|_{\infty}^{(M)}&=\|\mathbf{A}^{\intercal}\|_{1}^{(M)}\label{norm:5}\\
			\|\mathbf{A}\mathbf{B}\|_{\infty}^{(M)}&\leq\|\mathbf{A}\|_{\infty}^{(M)}\,\|\mathbf{B}\|_{\infty}^{(M)}\label{norm:6}
		\end{align}
	\end{subequations}
\end{lemma}	
\begin{proof}
	Proof of \eqref{norm:1}:
	\begin{align*}
		\|\mathbf{A}\|_{\max}^{(M)}=\max_{1\leq j,l\leq p}\|\mathbf{A}_{jl}\|_{F}=\max_{1\leq j,l\leq p}\|\text{vec}(\mathbf{A}_{jl})\|_{2}=\|\text{vec}(\mathbf{A})\|_{\max}^{(M^{2})}.
	\end{align*}
	
	\medskip
	Proof of \eqref{norm:2}: Note that for a matrix $\mathbf{G}=(g_{jl})$ and a vector $\mathbf{b}=(b_{j})$ of appropriate dimensions, 
	\begin{align*}
		\|\mathbf{Gb}\|_{2}^{2}&=\sum_{j}\sum_{l}(g_{jl}b_{l})^{2}\\
		&\leq\sum_{j}\left(\sum_{l}g_{jl}^{2}\right)\left(\sum_{l}b_{l}^{2}\right)\\
		&=\left(\sum_{j,l}g_{jl}^{2}\right)\left(\sum_{l}b_{l}^{2}\right)\\
		&=\|\mathbf{G}\|_{F}^{2}\,\|\mathbf{b}\|_{2}^{2}\\
		\Rightarrow \|\mathbf{Gb}\|_{2}&\leq \|\mathbf{G}\|_{F}\,\|\mathbf{b}\|_{2}.
	\end{align*}
	Therefore,
	\begin{align*}
		\|\mathbf{Au}\|_{\max}^{(M)}&=\max_{1\leq j,l\leq p}\|\mathbf{A}_{jl}\mathbf{u}_{l}\|_{2}\\
		&\leq \max_{1\leq j,l\leq p}\|\mathbf{A}_{jl}\|_{F}\,\|\mathbf{u}_{l}\|_{2}\\
		&\leq \left(\max_{1\leq j,l\leq p}\|\mathbf{A}_{jl}\|_{F}\right)\left(\max_{1\leq l\leq p}\|\mathbf{u}_{l}\|_{2}\right)\\
		&\leq \max_{1\leq j\leq p}\left(\sum_{l=1}^{p}\|\mathbf{A}_{jl}\|_{F}\right)\left(\max_{1\leq l\leq p}\|\mathbf{u}_{l}\|_{2}\right)\\
		&=\|\mathbf{A}\|_{\infty}^{(M)}\,\|\mathbf{u}\|_{\max}^{(M)}.
	\end{align*}
	
	\medskip
	Proof of \eqref{norm:3}:    
	\begin{align*}
		\|\mathbf{x}^{\intercal}\mathbf{y}\|_{F}&=\left\|\sum_{j=1}^{p}\mathbf{x}_{j}^{\intercal}\mathbf{y}_{j}\right\|_{F}\\
		&\leq \sum_{j=1}^{p}\|\mathbf{x}_{j}^{\intercal}\mathbf{y}_{j}\|_{F}\\
		&\leq \sum_{j=1}^{p}\|\mathbf{x}_{j}\|_{F}\,\|\mathbf{y}_{j}\|_{F}\\
		&\leq \left(\max_{1\leq l\leq p}\|\mathbf{x}_{l}\|_{F}\right)\sum_{j=1}^{p}\|\mathbf{y}_{j}\|_{F}\\
		&=\|\mathbf{x}\|_{\max}^{(M)}\,\|\mathbf{y}\|_{1}^{(M)}.
	\end{align*}

	\medskip
	Proof of \eqref{norm:4}:   
	\begin{align*}
		\|\mathbf{Ax}\|_{\max}^{(M)}&=\max_{1\leq j \leq p}\left\|\sum_{l=1}^{p}\mathbf{A}_{jl}\mathbf{x}_{l}\right\|_{F}\\
		&\leq \max_{1\leq j \leq p}\sum_{l=1}^{p}\|\mathbf{A}_{jl}\mathbf{x}_{l}\|_{F}\\
		&\leq \max_{1 \leq j \leq p}\sum_{l=1}^{p}\|\mathbf{A}_{jl}\|_{F}\,\|\mathbf{x}_{l}\|_{F}\\
		&=\max_{1\leq j \leq p}\sum_{l=1}^{p}\left(\max_{1\leq m \leq p}\|\mathbf{A}_{jm}\|_{F}\right)\|\mathbf{x}_{l}\|_{F}\\
		&= \left(\max_{1\leq j,m \leq p}\|\mathbf{A}_{jm}\|_{F}\right)\sum_{l=1}^{p}\|\mathbf{x}_{l}\|_{F}\\
		&=\|\mathbf{A}\|_{\max}^{(M)}\,\|\mathbf{x}\|_{1}^{(M)}.
	\end{align*}
	
	\medskip
	Proof of \eqref{norm:5}:  
	\begin{align*}
		\|\mathbf{A}^{\intercal}\|_{1}^{(M)}&=\max_{1\leq l\leq p}\sum_{j=1}^{p}\|(\mathbf{A}^{\intercal})_{jl}\|_{F}\\
		&=\max_{1\leq l\leq p}\sum_{j=1}^{p}\|\mathbf{A}_{lj}\|_{F}\\
		&=\max_{1\leq j \leq p}\sum_{l=1}^{p}\|\mathbf{A}_{jl}\|_{F}\\
		&=\|\mathbf{A}\|_{\infty}^{(M)}.
	\end{align*}
	
	\medskip
	Proof of \eqref{norm:6}: Note that
	\begin{align*}
		\|\mathbf{AB}\|_{F}^{2}&=\sum_{i,j}\left(\sum_{k}a_{ik}b_{kj}\right)^{2}\\
		&\leq\sum_{i,j}\left(\sum_{k}a_{ik}^{2}\right)\left(\sum_{k}b_{kj}^{2}\right)\\
		&=\left(\sum_{i,k}a_{ik}^{2}\right)\left(\sum_{l,j}b_{lj}^{2}\right)\\
		&=\|\mathbf{A}\|_{F}^{2}\,\|\mathbf{B}\|_{F}^{2}\\
		\Rightarrow \|\mathbf{AB}\|_{F}&\leq \|\mathbf{A}\|_{F}\,\|\mathbf{B}\|_{F}.
	\end{align*}
	Therefore,
	\begin{align*}
		\|\mathbf{AB}\|_{\infty}^{(M)}&=\max_{1\leq j\leq p}\sum_{l=1}^{p}\|(\mathbf{AB})_{jl}\|_{F}\\
		&=\max_{1\leq j \leq p}\sum_{l=1}^{p}\left\|\sum_{k=1}^{p}\mathbf{A}_{jk}\mathbf{B}_{kl}\right\|_{F}\\
		&\leq \max_{1\leq j\leq p}\sum_{l=1}^{p}\sum_{k=1}^{p}\|\mathbf{A}_{jk}\mathbf{B}_{kl}\|_{F}\\
		&\leq \max_{1\leq j\leq p}\sum_{l=1}^{p}\sum_{k=1}^{p}\|\mathbf{A}_{jk}\|_{F}\,\|\mathbf{B}_{kl}\|_{F}\\
		&=\max_{1\leq j\leq p}\left(\sum_{k=1}^{p}\|\mathbf{A}_{jk}\|_{F}\,\sum_{l=1}^{p}\|\mathbf{B}_{kl}\|_{F}\right)\\
		&\leq \max_{1\leq j\leq p}\left(\sum_{k=1}^{p}\|\mathbf{A}_{jk}\|_{F}\,\max_{1\leq m\leq p}\sum_{l=1}^{p}\|\mathbf{B}_{ml}\|_{F}\right)\\
		&=\|\mathbf{A}\|_{\infty}^{(M)}\,\|\mathbf{B}\|_{\infty}^{(M)}.
	\end{align*}
\end{proof}

\section{Proving Theorem~\ref{theorem:3}}\label{appendix:c}

Since the proof is the same for all $K$ graphs, $K$ is only present in the weights $\tau_{jl}$, and $K$ does not depend on $n$, we omit the superscript $k$ in this section. It is therefore implied that we are working within the $k$-th subpopulation for proving consistency.

Let $\mathbf{Z}=(\mathbf{Z}_{jl})$ be an element of the subdifferential $\partial(\sum_{j\ne l}\|\mathbf{\Omega}_{jl}\|_{F})/\partial{\mathbf{\Omega}}$, where
\begin{align}\label{subdiff}
	\mathbf{Z}_{jl}\in
	\begin{cases}
		\{\mathbf{0}\},&\quad\text{if}\quad j=l\\
		\left\{\frac{\mathbf{\Omega}_{jl}}{\|\mathbf{\Omega}_{jl}\|_{F}}\right\},&\quad\text{if}\quad j\ne l\quad\text{and}\quad\mathbf{\Omega}_{jl}\ne\mathbf{0}\\
		\{\mathbf{G}\in\mathbb{R}^{M\times M}\,:\,\|\mathbf{G}\|_{F}\leq 1\},&\quad\text{if}\quad j\ne l\quad\text{and}\quad\mathbf{\Omega}_{jl}=\mathbf{0}
	\end{cases}.
\end{align} 
Define the matrix of weights $\mathbf{T}=(\tau_{jl})$ with diagonal elements equal to 0, and the vector $\mathbf{1}_{M}=(1,\ldots,1)^{\intercal}\in\mathbb{R}^{M}$.

In the proof to follow $p_{n},M_{n},d_{n}$ will be abbreviated simply by $p,M,d$ and the $n$ subscript is going to be used only when it is meaningful. We begin by proving the existence of the solution of the adaptive fglasso problem
\begin{align}\label{adap.fglasso}
	\argmin_{\mathbf{\Omega}\succ\mathbf{0}}\left\{\tr(\hat{\mathbf{\Sigma}}\mathbf{\Omega})-\log\det(\mathbf{\Omega})+\lambda_{n}\sum_{j\ne l}\tau_{jl}\|\mathbf{\Omega}_{jl}\|_{F}\right\}
\end{align}
and providing optimality conditions for it.

\medskip
\begin{lemma}\label{lem:kkt}
	For any $\lambda_{n}>0$ and sample covariance matrix $\hat{\mathbf{\Sigma}}=(\hat{\sigma}_{jl})$ with strictly positive diagonal elements, the 
	adaptive fglasso problem \eqref{adap.fglasso} has a unique solution. Furthermore, this solution is equal to $\hat{\mathbf{\Omega}}$ if and only if
	\begin{align}\label{kkt}
		\hat{\mathbf{\Sigma}}-\hat{\mathbf{\Omega}}^{-1}+\lambda_{n}(\mathbf{T}\otimes\mathbf{1}_{M}\mathbf{1}_{M}^{\intercal})\circ\hat{\mathbf{Z}}=\mathbf{0},
	\end{align}
	where $\circ$ denotes the Hadamard product, and $\hat{\mathbf{Z}}$ is the subdifferential in \eqref{subdiff} evaluated at $\hat{\mathbf{\Omega}}$.
\end{lemma}
\begin{proof}
	By the Lagrangian duality, for $\lambda_{n}>0$, there is a constant $C(\lambda_{n})>0$ such that the problem \eqref{adap.fglasso} can be written in the equivalent constrained form
	\begin{align}\label{dual.adap.fglasso}
		\argmin_{\mathbf{\Omega}\in\mathcal{A}}\left\{\tr(\hat{\mathbf{\Sigma}}\mathbf{\Omega})-\log\det(\mathbf{\Omega})\right\},
	\end{align}
	where $\mathcal{A}=\{\mathbf{\Omega}\succ\mathbf{0}\,:\,\sum_{j\ne l}\tau_{jl}|\omega_{jl}|\leq C(\lambda_{n})\}$. It can be easily proved that the function
	\begin{align*}
		L:\mathcal{A}\rightarrow\mathbb{R},\quad L(\mathbf{\Omega})=\tr(\hat{\mathbf{\Sigma}}\mathbf{\Omega})-\log\det(\mathbf{\Omega})
	\end{align*} is convex \citep{boyd2004convex}. If $L$ is also bounded from below on its domain, then it has a unique minimum.
	
	Since the off-diagonal elements are bounded within an $\ell_{1}$-ball, the only possible issue is the behavior of the objective function on the diagonal elements. By Hadamard's inequality \citep[p.~35]{zhang2006schur}
	\begin{align*}
		\det(\mathbf{\Omega})\leq\prod_{i=1}^{pM}\omega_{ii}\Rightarrow \log\det(\mathbf{\Omega})\leq \sum_{i=1}^{pM}\log(\omega_{ii}).
	\end{align*}
	Thus,
	\begin{align*}
		\sum_{i=1}^{pM}\hat{\sigma}_{ii}\omega_{ii}-\log\det(\mathbf{\Omega})&\geq\sum_{i=1}^{pM}\left[\hat{\sigma}_{ii}\omega_{ii}-\log(\omega_{ii})\right]\geq pM[1+\log(\min_{i}(\hat{\sigma}_{ii})],
	\end{align*}
	which is bounded from below. Therefore, \eqref{dual.adap.fglasso} has a unique solution $\hat{\mathbf{\Omega}}$.
	
	By the interior extremum theorem \citep{spivak1980calculus}, if the global minimum of $L$ is achieved at $\hat{\mathbf{\Omega}}$, then the first derivative of $L$ at $\hat{\mathbf{\Omega}}$ will be zero, i.e.
	\begin{align*}
		\hat{\mathbf{\Sigma}}-\hat{\mathbf{\Omega}}^{-1}+\lambda_{n}(\mathbf{T}\otimes\mathbf{1}_{M}\mathbf{1}_{M}^{\intercal})\circ\hat{\mathbf{Z}}=\mathbf{0}
	\end{align*}
	The other direction is also true, since $L$ is a convex function.	
\end{proof}

Based on this lemma, we construct the primal-dual witness solution $(\tilde{\mathbf{\Omega}},\tilde{\mathbf{Z}})$ as follows:
\begin{enumerate}[label=(\alph*)]
	\item\label{step:a} We determine the matrix $\tilde{\mathbf{\Omega}}$ by solving the restricted adaptive fglasso
	\begin{align}\label{res.adap.fglasso}
		\tilde{\mathbf{\Omega}}=\argmin_{\mathbf{\Omega}\succ\mathbf{0},\mathbf{\Omega}_{\mathcal{S}^{c}}=\mathbf{0}}\left\{\tr(\hat{\mathbf{\Sigma}}\mathbf{\Omega})-\log\det(\mathbf{\Omega})+\lambda_{n}\sum_{j\ne l}\tau_{jl}\|\mathbf{\Omega}_{jl}\|_{F}\right\}.
	\end{align}
	
	\item\label{step:b} We choose $\tilde{\mathbf{Z}}_{\mathcal{S}}$ as a member of $\partial(\sum_{j\ne l}\|\tilde{\mathbf{\Omega}}_{jl}\|_{F})/\partial\mathbf{\Omega}_{\mathcal{S}}$.
	
	\item\label{step:c} For each $(j,l)\in\mathcal{S}^{c}$, we define 
	\begin{align*}
		\tilde{\mathbf{Z}}_{jl}:=\frac{1}{\lambda_{n}\tau_{jl}}\left[-\hat{\mathbf{\Sigma}}_{jl}+(\tilde{\mathbf{\Omega}}^{-1})_{jl}\right].
	\end{align*}
	
	\item\label{step:d} We verify the strict dual feasibility condition
	\begin{align*}
		\|\mathbf{Z}_{jl}\|_{F}<1\quad\text{for all}\quad (i,j)\in\mathcal{S}^{c}.
	\end{align*}
\end{enumerate}

\medskip
With step \ref{step:a} we ensure that
\begin{align*}
	\hat{\mathbf{\Sigma}}_{\mathcal{S}}-(\tilde{\mathbf{\Omega}}^{-1})_{\mathcal{S}}+\lambda_{n}(\mathbf{T}\otimes\mathbf{1}_{M}\mathbf{1}_{M}^{\intercal})_{\mathcal{S}}\circ\tilde{\mathbf{Z}}_{\mathcal{S}}=\mathbf{0},
\end{align*} 
which can be argued similarly as in Lemma \ref{lem:kkt}. The problem is that $\tilde{\mathbf{Z}}_{\mathcal{S}^{c}}$ is undefined, since in problem \eqref{res.adap.fglasso} we fix the elements $\mathbf{\Omega}_{\mathcal{S}^{c}}$ to be equal to zero. With step \ref{step:c} we define $\tilde{\mathbf{Z}}_{\mathcal{S}^{c}}$ so that $(\tilde{\mathbf{\Omega}},\tilde{\mathbf{Z}})$ is a solution of \eqref{kkt}. The only thing that remains to show is that $\tilde{\mathbf{Z}}_{\mathcal{S}^{c}}$ is an element of $\partial\left(\sum_{j\ne l}\|\mathbf{\tilde{\mathbf{\Omega}}}\|_{F}\right)/\partial\mathbf{\Omega}_{\mathcal{S}^{c}}$, which is the purpose of step \ref{step:d}. The result of the steps \ref{step:a}-\ref{step:d} is $\hat{\mathbf{\Omega}}=\tilde{\mathbf{\Omega}}$, which we need in order to show that $\hat{\mathcal{S}}\subset\mathcal{S}$.

In the analysis to follow, some additional notation is useful. We let $\mathbf{W}\in\mathbb{R}^{pM\times pM}$ denote the "effective noise" in the sample covariance matrix $\hat{\mathbf{\Sigma}}$--namely, the quantity
\begin{align}\label{effective noise}
	\mathbf{W}:=\hat{\mathbf{\Sigma}}-\mathbf{\Omega}_{0}^{-1}.
\end{align}
Second, we use $\mathbf{\Delta}=\tilde{\mathbf{\Omega}}-\mathbf{\Omega}_{0}$ to measure the discrepancy between the primal witness matrix $\tilde{\mathbf{\Omega}}$ and the truth $\mathbf{\Omega}_{0}$. Note that by the definition of $\tilde{\mathbf{\Omega}}$, $\mathbf{\Delta}_{\mathcal{S}^{c}}=\mathbf{0}$. Finally, we let $\mathbf{R}(\mathbf{\Delta})$ denote the difference of the gradient $\nabla (\log\det(\tilde{\mathbf{\Omega}}))$ from its first-order Taylor expansion around $\mathbf{\Omega}_{0}$. Using known results on the first and second derivatives of the log-determinant function \citep[p.~641]{boyd2004convex}, this remainder takes the form
\begin{align}\label{remainder}
	\mathbf{R}(\mathbf{\Delta})=\tilde{\mathbf{\Omega}}^{-1}-\mathbf{\Omega}_{0}^{-1}+\mathbf{\Omega}_{0}^{-1}\mathbf{\Delta}\mathbf{\Omega}_{0}^{-1}.
\end{align}

We begin by stating and proving a lemma that provides sufficient conditions for strict dual feasibility to hold, so that $\|\tilde{\mathbf{Z}}_{\mathcal{S}^{c}}\|_{\max}^{(M^{2})}<1$.

\medskip
\begin{lemma}[Strict dual feasibility]\label{lem:dual.feas.}
	Suppose that
	\begin{align*}
		\max\left\{\|\mathbf{W}\|_{\max}^{(M)},\|\mathbf{R}(\mathbf{\Delta})\|_{\max}^{(M)}\right\}\leq \frac{\lambda_{n}\left(a_{1}-a_{2} C_{\mathbf{\Gamma}^{2}}\right)}{2\left(1+ C_{\mathbf{\Gamma}^{2}}\right)},
	\end{align*}
	and
	\begin{align*}
		\min_{(j,l)\in\mathcal{S}^{c}}\tau_{jl}> a_{1},\quad\max_{(j,l)\in\mathcal{S}}\tau_{jl}< a_{2},
	\end{align*}
	for $a_{1},a_{2}$ specified in condition \ref{condition:weights}.	
	Then, the vector $\tilde{\mathbf{Z}}_{\mathcal{S}^{c}}$ constructed in step $(c)$ satisfies $\|\tilde{\mathbf{Z}}_{\mathcal{S}^{c}}\|_{\max}^{(M^{2})}<1$, and therefore $\tilde{\mathbf{\Omega}}=\hat{\mathbf{\Omega}}$.
\end{lemma}

\begin{proof}
	The optimality condition \eqref{kkt} can be rewritten in the alternative but equivalent form
	\begin{align}\label{kkt:alt}
		\mathbf{\Omega}_{0}^{-1}\mathbf{\Delta}\mathbf{\Omega}_{0}^{-1}+\mathbf{W}-\mathbf{R}(\mathbf{\Delta})+\lambda_{n}(\mathbf{T}\otimes\mathbf{1}_{M}\mathbf{1}_{M}^{\intercal})\circ\tilde{\mathbf{Z}}=\mathbf{0}.
	\end{align}
	The vectorized version of \eqref{kkt:alt} is
	\begin{align*}
		\mathbf{\Gamma}\vect(\mathbf{\Delta})+\vect(\mathbf{W})-\vect(\mathbf{R})+\lambda_{n}\vect(\mathbf{T}\otimes\mathbf{1}_{M}\mathbf{1}_{M}^{\intercal})\circ\vect(\tilde{\mathbf{Z}})=\mathbf{0},
	\end{align*}
	where we have abbreviated $\mathbf{R}(\mathbf{\Delta})$ by $\mathbf{R}$. Equivalently,
	\begin{align*}
		\left(
		\begin{array}{cc}
			\mathbf{\Gamma}_{\mathcal{S}\mathcal{S}}&\mathbf{\Gamma}_{\mathcal{S}\mathcal{S}^{c}}\\
			\mathbf{\Gamma}_{\mathcal{S}^{c}\mathcal{S}}&\mathbf{\Gamma}_{\mathcal{S}^{c}\mathcal{S}^{c}}
		\end{array}
		\right)
		\left(
		\begin{array}{c}
			\mathbf{\Delta}_{\mathcal{S}}\\
			\mathbf{\Delta}_{\mathcal{S}^{c}}
		\end{array}
		\right)
		+
		\left(
		\begin{array}{c}
			\mathbf{W}_{\mathcal{S}}-\mathbf{R}_{\mathcal{S}}+\lambda_{n}(\mathbf{T}\otimes\mathbf{1}_{M}\mathbf{1}_{M}^{\intercal})_{\mathcal{S}}\circ\tilde{\mathbf{Z}}_{\mathcal{S}}\\
			\mathbf{W}_{\mathcal{S}^{c}}-\mathbf{R}_{\mathcal{S}^{c}}+\lambda_{n}(\mathbf{T}\otimes\mathbf{1}_{M}\mathbf{1}_{M}^{\intercal})_{\mathcal{S}^{c}}\circ\tilde{\mathbf{Z}}_{\mathcal{S}^{c}}
		\end{array}
		\right)
		=\mathbf{0}.
	\end{align*}
	From this we get the system of equations
	\begin{subequations}
		\begin{align}
			\mathbf{\Gamma}_{\mathcal{S}\mathcal{S}}\mathbf{\Delta}_{\mathcal{S}}+\mathbf{W}_{\mathcal{S}}-\mathbf{R}_{\mathcal{S}}+\lambda_{n}(\mathbf{T}\otimes\mathbf{1}_{M}\mathbf{1}_{M}^{\intercal})_{\mathcal{S}}\circ\tilde{\mathbf{Z}}_{\mathcal{S}}&=\mathbf{0}\label{eq:1}\\
			\mathbf{\Gamma}_{\mathcal{S}^{c}\mathcal{S}}\mathbf{\Delta}_{\mathcal{S}}+\mathbf{W}_{\mathcal{S}^{c}}-\mathbf{R}_{\mathcal{S}^{c}}+\lambda_{n}(\mathbf{T}\otimes\mathbf{1}_{M}\mathbf{1}_{M}^{\intercal})_{\mathcal{S}^{c}}\circ\tilde{\mathbf{Z}}_{\mathcal{S}^{c}}&=\mathbf{0}\label{eq:2}
		\end{align}
	\end{subequations}
	Solving \eqref{eq:1} for $\mathbf{\Delta}_{S}$ and then substituting in \eqref{eq:2}, we get
	\begin{align*}
		\lambda_{n}(\mathbf{T}\otimes\mathbf{1}_{M}\mathbf{1}_{M}^{\intercal})_{\mathcal{S}^{c}}\circ\tilde{\mathbf{Z}}_{\mathcal{S}^{c}}&=\mathbf{R}_{\mathcal{S}^{c}}-\mathbf{W}_{\mathcal{S}^{c}}-\mathbf{\Gamma}_{\mathcal{S}^{c}\mathcal{S}}\mathbf{\Gamma}_{\mathcal{S}\mathcal{S}}^{-1}(\mathbf{R}_{\mathcal{S}}-\mathbf{W}_{\mathcal{S}})\\
		&+\lambda_{n}\mathbf{\Gamma}_{\mathcal{S}^{c}\mathcal{S}}\mathbf{\Gamma}_{\mathcal{S}\mathcal{S}}^{-1}(\mathbf{T}\otimes\mathbf{1}_{M}\mathbf{1}_{M}^{\intercal})_{\mathcal{S}}\circ\tilde{\mathbf{Z}}_{\mathcal{S}}.
	\end{align*}
	Taking the $M^{2}$-block versions of the $\ell_{\infty}$ norm on both sides, we have
	\begin{align*}
		\lambda_{n}\left(\min_{(j,l)\in\mathcal{S}^{c}}\tau_{jl}\right)\|\tilde{\mathbf{Z}}_{\mathcal{S}^{c}}\|_{\max}^{(M^{2})}\leq&\|\mathbf{R}_{\mathcal{S}^{c}}\|_{\max}^{(M^{2})}+\|\mathbf{W}_{\mathcal{S}^{c}}\|_{\max}^{(M^{2})}\\
		&+ C_{\mathbf{\Gamma}^{2}}\left(\|\mathbf{R}_{\mathcal{S}}\|_{\max}^{(M^{2})}+\|\mathbf{W}_{\mathcal{S}}\|_{\max}^{(M^{2})}\right)\\
		&+\lambda_{n} C_{\mathbf{\Gamma}^{2}}\left(\max_{(j,l)\in\mathcal{S}}\tau_{jl}\right).
	\end{align*}
	Using the condition of the lemma, we have
	\begin{align*}
		\|\tilde{\mathbf{Z}}_{\mathcal{S}^{c}}\|_{\max}^{(M^{2})}<\frac{1+ C_{\mathbf{\Gamma}^{2}}}{a_{1}\lambda_{n}}\left(\|\mathbf{W}\|_{\max}^{(M)}+\|\mathbf{R}(\mathbf{\Delta})\|_{\max}^{(M)}\right)+\frac{a_{2}}{a_{1}} C_{\mathbf{\Gamma}^{2}},
	\end{align*}
	which is no greater than 1.
\end{proof}

Our next step is to relate the behavior of the remainder term \eqref{remainder} to the deviation $\mathbf{\Delta}=\tilde{\mathbf{\Omega}}-\mathbf{\Omega}_{0}$.

\medskip
\begin{lemma}[Control of the remainder]\label{lem:remainder}
	Suppose that $\|\mathbf{\Delta}\|_{\max}^{(M)}\leq \frac{1}{3C_{\mathbf{\Sigma}}d}$ holds. Then the matrix
	\begin{align*}
		\mathbf{J}:=\sum_{s=0}^{\infty}(-1)^{s}(\mathbf{\Omega}_{0}^{-1}\mathbf{\Delta})^{s}
	\end{align*}
	satisfies $\|\mathbf{J}\|_{\infty}^{(M)}\leq\frac{3}{2}$. Moreover, the remainder $\mathbf{R}(\mathbf{\Delta})$ is equal to 
	\begin{align*}
		\mathbf{\Omega}_{0}^{-1}\mathbf{\Delta}\mathbf{\Omega}_{0}^{-1}\mathbf{\Delta}\mathbf{J}\mathbf{\Omega}_{0}^{-1}
	\end{align*}
	and has its $M$-block $\ell_{\infty}$ vector norm satisfying
	\begin{align}\label{R:max:norm:bound}
		\|\mathbf{R}(\mathbf{\Delta})\|_{\max}^{(M)}\leq\frac{3}{2}C_{\mathbf{\Sigma}}^{3}d\left(\|\mathbf{\Delta}\|_{\max}^{(M)}\right)^{2}.
	\end{align}
\end{lemma}	
\begin{proof}
	We write the remainder in the equivalent form
	\begin{align}\label{remainder:equivalent}
		\mathbf{R}(\mathbf{\Delta})=(\mathbf{\Omega}_{0}+\mathbf{\Delta})^{-1}-\mathbf{\Omega}_{0}^{-1}+\mathbf{\Omega}_{0}^{-1}\mathbf{\Delta}\mathbf{\Omega}_{0}^{-1}.
	\end{align}
	By the submultiplicativity of the $\|\cdot\|_{\infty}^{(M)}$, we have
	\begin{align}\label{geometric.series.bound}
		\|\mathbf{\Omega}_{0}^{-1}\mathbf{\Delta}\|_{\infty}^{(M)}\leq\|\mathbf{\Omega}_{0}^{-1}\|_{\infty}^{(M)}\|\mathbf{\Delta}\|_{\infty}^{(M)}\leq dC_{\mathbf{\Sigma}}\|\mathbf{\Delta}\|_{\max}^{(M)}<\frac{1}{3},
	\end{align}
	where we have used $\|\mathbf{\Delta}\|_{\max}^{(M)}\leq \frac{1}{3C_{\mathbf{\Sigma}}d}$, and the fact that for each $j$, $\mathbf{\Delta}$ has at most $d$ nonzero blocks $\mathbf{\Delta}_{jl}$. Consequently, we have the convergent matrix expansion \citep[p.~627]{schechter1996handbook}
	\begin{align}\label{matrix:expansion}
		(\mathbf{\Omega}_{0}+\mathbf{\Delta})^{-1}&=\left[\mathbf{\Omega}_{0}(\mathbf{I}+\mathbf{\Omega}_{0}^{-1}\mathbf{\Delta})\right]^{-1}\notag\\
		&=(\mathbf{I}+\mathbf{\Omega}_{0}^{-1}\mathbf{\Delta})^{-1}\mathbf{\Omega}_{0}^{-1}\notag\\
		&=\sum_{s=0}^{\infty}(-1)^{s}(\mathbf{\Omega}_{0}^{-1}\mathbf{\Delta})^{s}\mathbf{\Omega}_{0}^{-1}\notag\\
		&=\mathbf{\Omega}_{0}^{-1}-\mathbf{\Omega}_{0}^{-1}\mathbf{\Delta}\mathbf{\Omega}_{0}^{-1}+\mathbf{\Omega}_{0}^{-1}\mathbf{\Delta}\mathbf{\Omega}_{0}^{-1}\mathbf{\Delta}\mathbf{J}\mathbf{\Omega}_{0}^{-1}
	\end{align}
	Substituting \eqref{matrix:expansion} into \eqref{remainder:equivalent} yields
	\begin{align*}
		\mathbf{R}(\mathbf{\Delta})=\mathbf{\Omega}_{0}^{-1}\mathbf{\Delta}\mathbf{\Omega}_{0}^{-1}\mathbf{\Delta}\mathbf{J}\mathbf{\Omega}_{0}^{-1}.
	\end{align*}
	We now prove the bound on $\mathbf{R}(\mathbf{\Delta})$ as follows. Let the block matrix $\mathbf{e}_{j}\in\mathbb{R}^{pM\times M}$ with identity matrix in the $j$-th block and zero matrix elsewhere. Then,
	\begin{align*}
		\|\mathbf{R}(\mathbf{\Delta})\|_{\max}^{(M)}&=\max_{j,l}\|\mathbf{e}_{j}^{\intercal}\mathbf{\Omega}_{0}^{-1}\mathbf{\Delta}\mathbf{\Omega}_{0}^{-1}\mathbf{\Delta}\mathbf{J}\mathbf{\Omega}_{0}^{-1}\mathbf{e}_{l}\|_{F}\\
		&\leq\max_{j}\|\mathbf{\Delta}\mathbf{\Omega}_{0}^{-1}\mathbf{e}_{j}\|_{\max}^{(M)}\,\max_{l}\|\mathbf{\Omega}_{0}^{-1}\mathbf{\Delta}\mathbf{J}\mathbf{\Omega}_{0}^{-1}\mathbf{e}_{l}\|_{1}^{(M)},\quad(\text{by }\ref{norm:3})\\
		&\leq\max_{j}\|\mathbf{\Omega}_{0}^{-1}\mathbf{e}_{j}\|_{1}^{(M)}\,\|\mathbf{\Delta}\|_{\max}^{(M)}\,\max_{l}\|\mathbf{\Omega}_{0}^{-1}\mathbf{\Delta}\mathbf{J}\mathbf{\Omega}_{0}^{-1}\mathbf{e}_{l}\|_{1}^{(M)},\quad(\text{by }\ref{norm:4})\\
		&=\|\mathbf{\Omega}_{0}^{-1}\|_{\infty}^{(M)}\,\|\mathbf{\Delta}\|_{\max}^{(M)}\,\|\mathbf{\Omega}_{0}^{-1}\mathbf{J}^{\intercal}\mathbf{\Delta}\mathbf{\Omega}_{0}^{-1}\|_{\infty}^{(M)},\quad(\text{by }\ref{norm:5})\\
		&\leq C_{\mathbf{\Sigma}}^{3}\|\mathbf{\Delta}\|_{\max}^{(M)}\,
		\|\mathbf{J}^{\intercal}\|_{\infty}^{(M)}\,\|\mathbf{\Delta}\|_{\infty}^{(M)},\quad(\text{by }\ref{norm:6})\\
		&\leq C_{\mathbf{\Sigma}}^{3}d\left(\|\mathbf{\Delta}\|_{\max}^{(M)}\right)^{2}
		\|\mathbf{J}^{\intercal}\|_{\infty}^{(M)}.
	\end{align*} 
	Note that by \eqref{geometric.series.bound}, we have
	\begin{align*}
		\|\mathbf{J}^{\intercal}\|_{\infty}^{(M)}\leq\sum_{s=0}^{\infty}\left(\|\mathbf{\Delta}\mathbf{\Omega}_{0}^{-1}\|_{\infty}^{(M)}\right)^{s}\leq\frac{1}{1-\|\mathbf{\Delta}\|_{\infty}^{(M)}\|\mathbf{\Omega}_{0}^{-1}\|_{\infty}^{(M)}}<\frac{3}{2}.
	\end{align*}
	Thus, \eqref{R:max:norm:bound} holds.    
\end{proof}

Our next lemma provides control on the deviation $\mathbf{\Delta}=\tilde{\mathbf{\Omega}}-\mathbf{\Omega}_{0}$, measured in the $M$-block elementwise norm.

\medskip
\begin{lemma}[Control of $\mathbf{\Delta}$]\label{lem:diff.control}
	Suppose that
	\begin{align*}
		\max_{(j,l)\in\mathcal{S}}\tau_{jl}<a_{2}\quad\text{and}\quad r:=2 C_{\mathbf{\Gamma}}\left(\|\mathbf{W}\|_{\max}^{(M)}+\lambda_{n}a_{2}\right)\leq\min\left\{\frac{1}{3C_{\mathbf{\Sigma}}d},\frac{1}{3C_{\mathbf{\Sigma}}^{3} C_{\mathbf{\Gamma}}d}\right\}.
	\end{align*}
	Then, there exists $\mathbf{\Delta}\in\mathbb{R}^{pM\times pM}$ such that $\tilde{\mathbf{\Omega}}=\mathbf{\Delta}+\mathbf{\Omega}_{0}$ and $\|\mathbf{\Delta}\|_{\max}^{(M)}\leq r$.
\end{lemma}
\begin{proof}
	By arguing the same way as in Lemma \ref{lem:kkt}, there exists a unique solution to the restricted adaptive fglasso problem \eqref{res.adap.fglasso} and it is equal to $\tilde{\mathbf{\Omega}}$ if and only if $\tilde{\mathbf{\Omega}}_{\mathcal{S}^{c}}=\mathbf{0}$ and it is the root of the first-order derivative equation
	\begin{align*}
		\hat{\mathbf{\Sigma}}_{\mathcal{S}}-(\tilde{\mathbf{\Omega}}^{-1})_{\mathcal{S}}+\lambda_{n}(\mathbf{T}\otimes\mathbf{1}_{M}\mathbf{1}_{M}^{\intercal})_{\mathcal{S}}\circ\tilde{\mathbf{Z}}_{\mathcal{S}}=\mathbf{0},
	\end{align*}
	where $\tilde{\mathbf{Z}}_{\mathcal{S}}$ is a member of the subdifferential $\partial(\sum_{j=1}\|\tilde{\mathbf{\Omega}}_{jl}\|_{F})/\partial\mathbf{\Omega}_{\mathcal{S}}$.
	
	Let $\mathcal{A}=\{\mathbf{\Omega}\succ\mathbf{0}\,:\,\mathbf{\Omega}_{\mathcal{S}^{c}}=\mathbf{0}\}$ and $\mathcal{B}=\{\mathbf{\Delta}\in\mathbb{R}^{pM\times pM}\,:\,\mathbf{\Delta}_{\mathcal{S}^{c}}=\mathbf{0}\}$.  We define the functions $\mathbf{G}:\mathcal{A}\rightarrow\mathbb{R}^{|\mathcal{S}|^{2}M^{2}}$and $\mathbf{F}:\mathcal{B}\rightarrow\mathbb{R}^{|\mathcal{S}|^{2}M^{2}}$ as
	\begin{align}
		\mathbf{G}(\mathbf{\Omega})&=-[\mathbf{\Omega}^{-1}]_{\mathcal{S}}+\hat{\mathbf{\Sigma}}_{\mathcal{S}}+\lambda_{n}(\mathbf{T}\otimes\mathbf{1}_{M}\mathbf{1}_{M}^{\intercal})_{\mathcal{S}}\circ\tilde{\mathbf{Z}}_{\mathcal{S}},\notag\\
		\mathbf{F}(\mathbf{\Delta})&=-\mathbf{\Gamma_{\mathcal{S}\mathcal{S}}}^{-1}\mathbf{G}(\mathbf{\Omega}_{0}+\mathbf{\Delta})+\mathbf{\Delta}_{\mathcal{S}}.\label{F:eq}
	\end{align}
	Finally, we define the closed ball
	\begin{align*}
		\mathcal{C}(\mathcal{B},r)=\left\{\mathbf{\Delta}\in\mathcal{B}\,:\,\|\mathbf{\Delta}_{\mathcal{S}}\|_{\max}^{(M^2)}\leq r\right\}.
	\end{align*}
	Observe that
	\begin{align}
		\mathbf{G}(\mathbf{\Omega}_{0}+\mathbf{\Delta})&=-[(\mathbf{\Omega}_{0}+\mathbf{\Delta})^{-1}]_{\mathcal{S}}+\hat{\mathbf{\Sigma}}_{\mathcal{S}}+\lambda_{n}(\mathbf{T}\otimes\mathbf{1}_{M}\mathbf{1}_{M}^{\intercal})_{\mathcal{S}}\circ\tilde{\mathbf{Z}}_{\mathcal{S}}\notag\\
		&=-[\mathbf{R}(\mathbf{\Delta})-\mathbf{\Omega}_{0}^{-1}\mathbf{\Delta}\mathbf{\Omega}_{0}^{-1}]_{\mathcal{S}}+\mathbf{W}_{\mathcal{S}}+\lambda_{n}(\mathbf{T}\otimes\mathbf{1}_{M}\mathbf{1}_{M}^{\intercal})_{\mathcal{S}}\circ\tilde{\mathbf{Z}}_{\mathcal{S}}\notag\\
		&=-\mathbf{R}_{\mathcal{S}}+\mathbf{\Gamma}_{\mathcal{S}\mathcal{S}}\mathbf{\Delta}_{\mathcal{S}}+\mathbf{W}_{\mathcal{S}}+\lambda_{n}(\mathbf{T}\otimes\mathbf{1}_{M}\mathbf{1}_{M}^{\intercal})_{\mathcal{S}}\circ\tilde{\mathbf{Z}}_{\mathcal{S}}.\label{G:eq}
	\end{align}
	Substituting \eqref{G:eq} into \eqref{F:eq}, we obtain
	\begin{align*}
		\mathbf{F}(\mathbf{\Delta})=\underbrace{\mathbf{\Gamma}_{\mathcal{S}\mathcal{S}}^{-1}\mathbf{R}_{\mathcal{S}}}_{\mathbf{T}_{1}}-\underbrace{\mathbf{\Gamma}_{\mathcal{S}\mathcal{S}}^{-1}[\mathbf{W}_{\mathcal{S}}+\lambda_{n}(\mathbf{T}\otimes\mathbf{1}_{M}\mathbf{1}_{M}^{\intercal})_{\mathcal{S}}\circ\tilde{\mathbf{Z}}_{\mathcal{S}}]}_{\mathbf{T}_{2}}.
	\end{align*}
	
	Let $\mathbf{\Delta}\in\mathcal{C}(\mathcal{B},r)$. Then, by Lemma \ref{lem:remainder} and the conditions of this lemma about $r$, we get
	\begin{align*}
		\|\mathbf{T}_{1}\|_{\max}^{(M^{2})}\leq C_{\mathbf{\Gamma}}\|\mathbf{R}_{\mathcal{S}}\|_{\max}^{(M^{2})}\leq  C_{\mathbf{\Gamma}}\frac{3}{2}C_{\mathbf{\Sigma}}^{3}dr^{2}\leq\frac{r}{2}.
	\end{align*}
	Concerning $\mathbf{T}_{2}$, we have
	\begin{align*}
		\|\mathbf{T}_{2}\|_{\max}^{(M^{2})}\leq C_{\mathbf{\Gamma}}\left(\|\mathbf{W}\|_{\max}^{(M)}+\lambda_{n}\max_{(j,l)\in\mathcal{S}}\tau_{jl}\right)\leq\frac{r}{2}.
	\end{align*}
	Combining the two bounds, we get $\|\mathbf{F}(\mathbf{\Delta})\|\leq r$ for every $\mathbf{\Delta}\in\mathcal{C}(\mathcal{B},r)$, which means that $\mathbf{F}(\mathcal{C}(\mathcal{B},r))\subset \mathcal{C}(\mathcal{B},r)$. By Brouwer's fixed point theorem \citep[p.~161]{ortega2000iterative}, there exists $\mathbf{\Delta}\in\mathcal{C}(\mathcal{B},r)$ such that
	\begin{align*}
		\mathbf{F}(\mathbf{\Delta})=\mathbf{\Delta}\Leftrightarrow \mathbf{G}(\mathbf{\Omega}_{0}+\mathbf{\Delta})=\mathbf{0}\Leftrightarrow\tilde{\mathbf{\Omega}}=\mathbf{\Omega}_{0}+\mathbf{\Delta},
	\end{align*}
	which completes the proof.
\end{proof}

Let $\mathbb{N}$ denote the set of positive integers and let the true covariance matrix be $\mathbf{\Sigma}_{0}=(\sigma_{0,jl})$. To prove the tail bounds of $\|\mathbf{W}\|_{\max}^{(M)}$ and $\|\mathbf{R}(\mathbf{\Delta})\|_{\max}^{(M)}$, we make use of the following useful definition.

\medskip
\begin{definition}[Tail condition]
	The random vector $\mathbf{X}\in\mathbb{R}^{p_{n}M_{n}}$ satisfies the tail condition $\mathcal{T}(f,\nu_{\star})$ if there exists a constant $\nu_{\star}\in(0,\infty]$ and a function $f:\mathbb{N}\times(0,\infty)\rightarrow(0,\infty)$ such that for any $i,j=1,\ldots,p_{n}M_{n}$, we have
	\begin{align*}
		P(|\hat{\sigma}_{ij}-\sigma_{0,ij}|\geq\delta)\leq\frac{1}{f(n,\delta)},\quad\text{for all }\delta\in(0,1/\nu_{\star}].
	\end{align*}
\end{definition}
Given a larger sample size $n$, we expect the tail probability bound $1/f(n,\delta)$ to be smaller, or equivalently, for the tail function $f(n,\delta)$ to be larger. Accordingly, we require that $f$ is monotonically increasing in $n$, so that for each fixed $\delta>0$, we can define the inverse function
\begin{align}\label{bar:n}
	\bar{n}_{f}(\delta,r)&=\argmax\{n:f(n,\delta)\leq r\}.
\end{align}
Similarly, we expect that $f$ is monotonically increasing in $\delta$, so that for each fixed $n$, we can define the inverse in the second argument
\begin{align}\label{bar:delta}
	\bar{\delta}_{f}(n,r)&=\argmax\{\delta:f(n,\delta)\leq r\},
\end{align}
where $r\in[1,\infty)$. For future reference, we note a simple consequence of the monotonicity of the tail function $f$, that is
\begin{align}\label{n:delta}
	n>\bar{n}_{f}(\delta,r)\quad\text{for some }\delta>0\quad\Longrightarrow\quad \bar{\delta}_{f}(n,r)<\delta.
\end{align}

It can be proven \citep[Theorem 1; ][]{qiao2019functional} that under condition \ref{condition:concetration}, one such tail function is
\begin{align*}
	f(n,\delta)=\frac{1}{C_{2}}\exp\left\{C_{1}n^{1-2\alpha(1+\beta)}\delta^{2}\right\},
\end{align*}
for some constants $C_{1},C_{2}>0$ and any $0<\delta\leq C_{1}$. Define
\begin{align*}
	\delta_{1}=\frac{\min\left\{\frac{1}{3C_{\mathbf{\Sigma}}d_{n}},\frac{1}{3C_{\mathbf{\Sigma}}^{3} C_{\mathbf{\Gamma}}d_{n}}\right\}}{2M_{n} C_{\mathbf{\Gamma}}\left[1+\frac{2a_{2}(1+ C_{\mathbf{\Gamma}^{2}})}{a_{1}-a_{2} C_{\mathbf{\Gamma}^{2}}}\right]},\quad
	\delta_{2}=\frac{1}{6M_{n}d_{n}C_{\mathbf{\Sigma}}^{3} C_{\mathbf{\Gamma}}^{2}\left[1+\frac{2a_{2}(1+ C_{\mathbf{\Gamma}^{2}})}{a_{1}-a_{2} C_{\mathbf{\Gamma}^{2}}}\right]^{2}}
\end{align*}

\medskip
\begin{lemma}\label{lem:concetration:remainder}
	Suppose that
	\begin{align*}
		\|\mathbf{W}\|_{\max}^{(M)}\leq \frac{\lambda_{n}\left(a_{1}-a_{2} C_{\mathbf{\Gamma}^{2}}\right)}{2\left(1+ C_{\mathbf{\Gamma}^{2}}\right)}\quad\text{and}\quad\max_{(j,l)\in\mathcal{S}}\tau_{jl}<a_{2},
	\end{align*}
	for $a_{1},a_{2}$ specified in condition \ref{condition:weights}. Then, for any $\gamma>2$,
	\begin{align*}
		n>\bar{n}_{f}\left(\min\{C_{1},\delta_{1},\delta_{2}\},(M_{n}p_{n})^{\gamma}\right),\quad
		\lambda_{n}=\frac{2(1+ C_{\mathbf{\Gamma}^{2}})M_{n}\bar{\delta}_{f}\left(n,(M_{n}p_{n})^{\gamma}\right)}{a_{1}-a_{2} C_{\mathbf{\Gamma}^{2}}},
	\end{align*}
	we have,
	\begin{align*}
		\|\mathbf{R}(\mathbf{\Delta})\|_{\max}^{(M)}\leq \frac{\lambda_{n}\left(a_{1}-a_{2} C_{\mathbf{\Gamma}^{2}}\right)}{2\left(1+ C_{\mathbf{\Gamma}^{2}}\right)}.
	\end{align*}
\end{lemma}
\begin{proof}
	Observe that
	\begin{align*}
		2 C_{\mathbf{\Gamma}}(\|\mathbf{W}\|_{\max}^{(M)}+\lambda_{n}a_{2})\leq 2 C_{\mathbf{\Gamma}}\left[1+\frac{2a_{2}(1+ C_{\mathbf{\Gamma}^{2}})}{a_{1}-a_{2} C_{\mathbf{\Gamma}^{2}}}\right]M_{n}\bar{\delta}_{f}\left(n,(M_{n}p_{n})^{\gamma}\right).
	\end{align*}
	From the lower bound on $n$ and the monotonicity of $f$, we have
	\begin{align*}
		2 C_{\mathbf{\Gamma}}\left[1+\frac{2a_{2}(1+ C_{\mathbf{\Gamma}^{2}})}{a_{1}-a_{2} C_{\mathbf{\Gamma}^{2}}}\right]M_{n}\bar{\delta}_{f}\left(n,(M_{n}p_{n})^{\gamma}\right)\leq\min\left\{\frac{1}{3C_{\mathbf{\Sigma}}d_{n}},\frac{1}{3C_{\mathbf{\Sigma}}^{3} C_{\mathbf{\Gamma}}d_{n}}\right\}.
	\end{align*}
	By Lemmas \ref{lem:remainder} and \ref{lem:diff.control}, we have
	\begin{align*}
		\|\mathbf{R}(\mathbf{\Delta})\|_{\max}^{(M)}
		&\leq6d_{n}C_{\mathbf{\Sigma}}^{3} C_{\mathbf{\Gamma}}^{2}\left[1+\frac{2a_{2}(1+ C_{\mathbf{\Gamma}^{2}})}{a_{1}-a_{2} C_{\mathbf{\Gamma}^{2}}}\right]^{2}M_{n}^{2}\bar{\delta}_{f}\left(n,(M_{n}p_{n})^{\gamma}\right)^{2}\\
		&\leq M_{n}\bar{\delta}_{f}\left(n,(M_{n}p_{n})^{\gamma}\right)=\frac{\lambda_{n}\left(a_{1}-a_{2} C_{\mathbf{\Gamma}^{2}}\right)}{2\left(1+ C_{\mathbf{\Gamma}^{2}}\right)},
	\end{align*}
	where the last inequality follows from the lower bound on the sample size $n$ and the monotonicity of $f$.
\end{proof}

\medskip
\begin{lemma}\label{lem:concentration:noise}
	Suppose condition \ref{condition:concetration} holds. For any $\gamma>2$, $n>\bar{n}_{f}\left(C_{1},(M_{n}p_{n})^{\gamma}\right)$ and
	\begin{align*}
		\lambda_{n}=\frac{2(1+ C_{\mathbf{\Gamma}^{2}})M_{n}\bar{\delta}_{f}\left(n,(M_{n}p_{n})^{\gamma}\right)}{a_{1}-a_{2} C_{\mathbf{\Gamma}^{2}}},
	\end{align*}
	we have,
	\begin{align*}
		\|\mathbf{W}\|_{\max}^{(M)}\leq \frac{\lambda_{n}\left(a_{1}-a_{2} C_{\mathbf{\Gamma}^{2}}\right)}{2\left(1+ C_{\mathbf{\Gamma}^{2}}\right)},
	\end{align*}
	with probability greater than $1-(M_{n}p_{n})^{2-\gamma}$.
\end{lemma}
\begin{proof}
	By Lemma 14 of \cite{qiao2019functional}, we have
	\begin{align*}
		P\left(\|\mathbf{W}\|_{\max}^{(M)}\geq M_{n}\bar{\delta}_{f}\left(n,(M_{n}p_{n})^{\gamma}\right)\right)\leq(Mp)^{2-\gamma}.
	\end{align*}
	Thus, for
	\begin{align*}
		\lambda_{n}=\frac{2(1+ C_{\mathbf{\Gamma}^{2}})M_{n}\bar{\delta}_{f}\left(n,(M_{n}p_{n})^{\gamma}\right)}{a_{1}-a_{2} C_{\mathbf{\Gamma}^{2}}},
	\end{align*}
	we have
	\begin{align*}
		\|\mathbf{W}\|_{\max}^{(M)}\leq \frac{\lambda_{n}\left(a_{1}-a_{2} C_{\mathbf{\Gamma}^{2}}\right)}{2\left(1+ C_{\mathbf{\Gamma}^{2}}\right)},
	\end{align*}
	with probability greater than $1-(M_{n}p_{n})^{2-\gamma}$. 
\end{proof}

We are now ready to prove graph selection consistency.

\medskip
\begin{theorem}
	Suppose conditions \ref{condition:concetration} and \ref{condition:weights} hold, $\gamma>2$,
	\begin{align*}
		\lambda_{n}=\frac{2(1+ C_{\mathbf{\Gamma}^{2}})M_{n}}{a_{1}-a_{2} C_{\mathbf{\Gamma}^{2}}}\sqrt{\frac{\log C_{2}+\gamma\log(M_{n}p_{n})}{C_{1}n^{1-2\alpha(1+\beta)}}},
	\end{align*}
	and    
	\begin{align*}
		\min_{(j,l)\in\mathcal{S}}\|\mathbf{\Omega}_{0,jl}\|_{F}>\min\left\{\frac{1}{3C_{\mathbf{\Sigma}}d_{n}},\frac{1}{3C_{\mathbf{\Sigma}}^{3} C_{\mathbf{\Gamma}}d_{n}}\right\}.
	\end{align*}
	Then, for all $n$ satisfying the lower bound
	\begin{scriptsize}
		\begin{align*}
			n^{1-2\alpha(1+\beta)}>\frac{\log[C_{2}(M_{n}p_{n})^{\gamma}]}{C_{1}}\max\left\{\frac{1}{C_{1}},\frac{2M_{n} C_{\mathbf{\Gamma}}\left[1+\frac{2a_{2}(1+ C_{\mathbf{\Gamma}^{2}})}{a_{1}-a_{2} C_{\mathbf{\Gamma}^{2}}}\right]}{\min\left\{\frac{1}{3C_{\mathbf{\Sigma}}d_{n}},\frac{1}{3C_{\mathbf{\Sigma}}^{3} C_{\mathbf{\Gamma}}d_{n}}\right\}},6M_{n}d_{n} C_{\mathbf{\Gamma}}^{2} C_{\mathbf{\Gamma}^{2}}\left[1+\frac{2a_{2}(1+ C_{\mathbf{\Gamma}^{2}})}{a_{1}-a_{2} C_{\mathbf{\Gamma}^{2}}}\right]^{2}\right\}^{2},
		\end{align*}
	\end{scriptsize}
	we have $\hat{\mathcal{S}}=\mathcal{S}$ with probability greater than $1-3(M_{n}p_{n})^{2-\gamma}$.
\end{theorem}
\begin{proof}
	By Lemmas \ref{lem:dual.feas.} and \ref{lem:concetration:remainder} we can see that
	\begin{align*}
		\left\{\hat{\mathbf{\Omega}}=\tilde{\mathbf{\Omega}}\right\}\supset&\left\{\max_{(j,l)\in\mathcal{S}}\tau_{jl}<a_{2}\right\}\cap\left\{\min_{(j,l)\in\mathcal{S}^{c}}\tau_{jl}>a_{1}\right\}\\
		&\cap\left\{\|\mathbf{W}\|_{\max}^{(M)}\leq \frac{\lambda_{n}\left(a_{1}-a_{2} C_{\mathbf{\Gamma}^{2}}\right)}{2\left(1+ C_{\mathbf{\Gamma}^{2}}\right)}\right\}.
	\end{align*}
	Hence, by Lemmas \ref{lem:dual.feas.}, \ref{lem:concetration:remainder} and \ref{lem:concentration:noise}, we have
	\begin{align}\label{prob:lower:bound}
		P\left(\hat{\mathbf{\Omega}}=\tilde{\mathbf{\Omega}}\right)\geq 1-3(M_{n}p_{n})^{2-\gamma}.
	\end{align}
	To find the lower bound of $n$ so that \eqref{prob:lower:bound} is satisfied we use Definition \ref{bar:n}, with $\delta=\min\{C_{1},\delta_{1},\delta_{2}\}$ and $r=(M_{n}p_{n})^{\gamma}$. Then we fix an $n>\bar{n}_{f}(\min\{C_{1},\delta_{1},\delta_{2}\},(M_{n}p_{n})^{\gamma})$ and find $\bar{\delta}_{f}\left(n,(M_{n}p_{n})^{\gamma}\right)$ using Definition \ref{bar:delta}, relationship \eqref{n:delta}, and substituting in $\lambda_{n}$.
	
	Conditioning on the event $\hat{\mathbf{\Omega}}=\tilde{\mathbf{\Omega}}$, we have $\mathcal{S}^{c}\subset\hat{\mathcal{S}}^{c}$. By Lemma \ref{lem:diff.control}, for any $(j,l)\in\hat{\mathcal{S}}^{c}\cap\mathcal{S}$, we have
	\begin{align*}
		\|\mathbf{\Omega}_{0,jl}\|_{F}=\|\mathbf{\Omega}_{0,jl}-\hat{\mathbf{\Omega}}_{jl}\|_{F}=\|\mathbf{\Omega}_{0,jl}-\tilde{\mathbf{\Omega}}_{jl}\|_{F}=\|\mathbf{\Delta}_{jl}\|\leq \min\left\{\frac{1}{3C_{\mathbf{\Sigma}}d_{n}},\frac{1}{3C_{\mathbf{\Sigma}}^{3} C_{\mathbf{\Gamma}}d_{n}}\right\},
	\end{align*}
	which is a contradiction. Thus, $\hat{\mathcal{S}}=\mathcal{S}$.
\end{proof}

\newpage
\nocite{*}
\bibliographystyle{asa}
\bibliography{references}

\end{document}